%% file: main.tex
\documentclass[a4paper]{article}

\usepackage{ltexpprt}
\newcommand*\samethanks[1][\value{footnote}]{\footnotemark[#1]}
\usepackage{comment}

\usepackage{mathtools}
\usepackage{amsfonts,amssymb,amsmath}
\usepackage{algorithm}
\usepackage{algorithmic}
\usepackage{multirow}

\usepackage{xcolor}
\usepackage{float}

\DeclareMathOperator{\E}{{E}}
\DeclareMathOperator*{\argmax}{arg\,max}

\newcommand{\xue}[1]{\textcolor{red}{\bf (xue: #1)}}
\title{Near Optimal Jointly Private Packing Algorithms via Dual Multiplicative Weight Update}
\author{Zhiyi Huang\thanks{Department of Computer Science, the University of Hong Kong. This work is supported in by a RGC grant HKU27200214E. {\texttt{\{zhiyi,xzhu2\}@cs.hku.hk}}}
	\and Xue Zhu\samethanks}
\date{}

\begin{document}

\maketitle

\input{abstract}

\input{introduction}

\input{preliminaries}

\input{algorithm}

\input{hardness}

\input{online}

\bibliographystyle{plain}
\bibliography{main}
\appendix
\input{appendix}

\end{document}

%% file: abstract.tex
\begin{abstract}
We present an improved $(\epsilon, \delta)$-jointly differentially private algorithm for packing problems.
Our algorithm gives a feasible output that is approximately optimal up to an $\alpha n$ additive factor as long as the supply of each resource is at least $\tilde{O}(\sqrt{m} / \alpha \epsilon)$, where $m$ is the number of resources.
This improves the previous result by Hsu et al.~(SODA '16), which requires the total supply to be at least $\tilde{O}(m^2 / \alpha \epsilon)$, and only guarantees approximate feasibility in terms of total violation.
Further, we complement our algorithm with an almost matching hardness result, showing that $\Omega(\sqrt{m \ln(1/\delta)} / \alpha \epsilon)$ supply is necessary for any $(\epsilon, \delta)$-jointly differentially private algorithm to compute an approximately optimal packing solution.
Finally, we introduce an alternative approach that runs in linear time, is exactly truthful, can be implemented online, and can be $\epsilon$-jointly differentially private, but requires a larger supply of each resource.

\end{abstract}

%% file: introduction.tex
\section{Introduction}
\label{sec:intro}

Handling user data has become the focal point of modern computational problems, bringing up many new challenges including user privacy.
The consensus notion of privacy in theoretical computer science is differential privacy introduced by Dwork et al.~\cite{DBLP:conf/tcc/DworkMNS06}.
Informally, a mechanism is differentially private if the output distribution is insensitive to the change of an individual's data.
Since its introduction, the community has introduced differentially private algorithms for a variety of problems, including computation of numerical statistics (e.g., \cite{dwork2010differential, DBLP:journals/jacm/BlumLR13}), machine learning (e.g., \cite{DBLP:journals/jmlr/ChaudhuriH11, kasiviswanathan2011can}, game theory (e.g., \cite{DBLP:conf/focs/McSherryT07, DBLP:reference/algo/Huang16c}), etc.

However, many fundamental combinatorial problems, such as the bipartite matching problem, provably do not admit any differentially private algorithm with non-trivial approximation guarantee~\cite{DBLP:conf/stoc/HsuHRRW14}.
To resolve the situation, Hsu et al.~\cite{DBLP:conf/stoc/HsuHRRW14} adopt a relaxed notion called joint differential privacy by Kearns et al.~\cite{DBLP:conf/innovations/KearnsPRU14} and introduce jointly differentially private algorithms for the bipartite matching problem and more generally the welfare maximization problem w.r.t.\ gross substitutes valuations.
Hsu et al.~\cite{DBLP:conf/soda/Hsu0RW16} further develop the techniques in~\cite{DBLP:conf/stoc/HsuHRRW14}, and propose a general framework that solves a large family of convex programs in a jointly differentially private manner via a noisy version of the dual gradient descent method.

While our techniques can be applied to the generic convex programs studied in \cite{DBLP:conf/soda/Hsu0RW16}, we focus on the packing problem in this conference version due to space constraint.
Consider a packing problem with $n$ agents and $m$ resources where each agent has different values for different bundles of resources.
The value and resource demands of an agent are private.
The goal is to allocate bundles to agents so as to maximize the sum of values of all agents subject to the supply constrains of resources.
The algorithm by Hsu et al.~\cite{DBLP:conf/soda/Hsu0RW16} incurs an $\tilde{O}(m^2 / \epsilon)$ additive loss in terms of the objective and up to $\tilde{O}(m^2 / \epsilon)$ additive total violation of the supply constraints.
In other words, their algorithm guarantees up to $\alpha n$ additive loss in the objective if $n \ge \tilde{O}(m^2 / \alpha \epsilon)$, and the total violation is at most $\alpha$ fraction of the total supply if the supply per constraint is at least $\tilde{O}(m / \alpha \epsilon)$.
No non-trivial guarantee is given in terms of per constraint violation.

\medskip

\textbf{Our contributions.}
Our first result is an $(\epsilon, \delta)$-jointly differentially private algorithm (Sec.~\ref{sec:algorithm}) that improves the results by Hsu et al.~\cite{DBLP:conf/soda/Hsu0RW16} by two means:
(1) It reduces the supply requirement in terms of the dependence on the number of resources $m$; and (2) it provides exact feasibility with high probability.
\footnote{We note that Hsu et al.~\cite{DBLP:conf/soda/Hsu0RW16} can also obtain exact feasibility by shifting the error to the objective, under the extra assumption that each unit of resource provides at most value $1$ in the objective. Our result does not need such an assumption.}
Concretely, we show the following:

\begin{theorem}
\label{thm:algorithm}
For any packing problem with $m$ constraints, there is an $(\epsilon, \delta)$-jointly differentially private algorithm whose output is feasible and approximately optimal up to an $\alpha n$ additive factor as long as the supply of each resource is at least $\tilde{O}( \sqrt{m} / \alpha \epsilon)$.
\end{theorem}

Our algorithm is a noisy version of the multiplicative weight update algorithm running in the dual space.
To compute a primal packing solution, we maintain a set of dual prices that coordinates the primal decisions: an item is selected in the packing if and only if its value is greater than the total price of its resource demands.
The dual prices are then updated using the multiplicative weight update method (e.g., \cite{DBLP:conf/soda/AgrawalD15}) with Laplacian noise added to each iteration to ensure that the dual prices are differentially private.
Finally, the billboard lemma introduced by Hsu et al.~\cite{DBLP:conf/stoc/HsuHRRW14} indicates that the primal packing solution satisfies joint differential privacy if the coordinators, i.e., the dual prices, satisfy differential privacy.

The main technical challenges then arise from bounding the error introduced by the Laplacian noise.
First, the standard analysis of the multiplicative weight update framework requires the update weights to be bounded, while our noisy update weights are unbounded.
To resolve this problem, we truncate noisy update weights that are either too large or too small, and bound the error due to this truncation using the small-tail properties of Laplacian distributions (Lemma ~\ref{lem:5}).
Second, to obtain with-high-probability guarantees, one resorts to concentration bounds, which generally require the random variables to be bounded.
To this end, we introduce a concentration lemma for Laplacian distributed variables (Lemma~\ref{lem:4}).

Then, we complement our algorithm with an almost matching lower bound on the minimum supply required for any $(\epsilon, \delta)$-jointly differentially private algorithm to compute an approximately optimal solution (Sec.~\ref{sec:hardness}).

\begin{theorem}
\label{thm:hardness}
(Hardness for $(\epsilon, \delta)$-private
 algorithms) If there is an $(\epsilon, \delta)$-jointly differentially private algorithm that with high probability outputs a feasible solution for the packing problem that is approximately optimal up to an additive $O(\alpha n)$, then the supply per constraint must be at least
\[
\Omega \left( \tfrac{\sqrt{m \ln(1/\delta)}}{\alpha \epsilon} \right) ~.
\]
\end{theorem}

This lower bound significantly improves the best previous bound of $\Omega\big({1}/{\sqrt{\alpha}}\big)$ by Hsu et al.~\cite{DBLP:conf/stoc/HsuHRRW14}.
It matches the supply requirement in Theorem~\ref{thm:algorithm} up to log factors, certifying that our algorithm has achieved near optimal trade-offs between privacy and accuracy.
The proof of this hardness result draws a novel connection between jointly differentially private packing algorithms and differentially private query release algorithms.
We show that one can take an arbitrary jointly differentially private packing algorithm as a blackbox and use it to construct a differentially private query release algorithm for answering arbitrary counting queries.
Then, the hardness result follows from existing hardness for query release by Steinke and Ullman~\cite{steinke2017between}.

Having achieved the optimal tradeoffs between privacy and accuracy, we turn to other aspects of the algorithm such as running time.
A common drawback of the algorithms in the previous dual gradient descent approach by Hsu et al.~\cite{DBLP:conf/soda/Hsu0RW16} and that in Theorem~\ref{thm:algorithm} is the cubic dependence in $n$ in the running time.
\emph{Is there a privacy-preserving algorithm that goes over each agent's data only once and still finds an approximately optimal packing solution?}

To this end, we introduce an alternative approach that can be interpreted as a privacy-preserving version of the online packing algorithm in the random-arrival model by Agrawal and Devanur~\cite{DBLP:conf/soda/AgrawalD15}.
In each round of dual update, instead of computing the best responses of all agents to calculate an accurate dual subgradient, the new approach picks one agent and uses his best response to calculate a proxy subgradient and updates the dual prices accordingly; the agent gets his best response bundle.
The algorithm updates the dual prices for exactly $n$ rounds, using each agent to compute the proxy subgradient exactly once.
The ordering that the algorithm picks the agent is chosen uniformly at random at the beginning.

The new approach runs in linear time, significantly faster than the other approaches whose running time has cubic dependence in $n$.
However, it needs a larger supply of each resource, i.e., at least $\tilde{O}\big( \frac{\sqrt{nm}}{\alpha \epsilon} \big)$, in order to compute an approximately optimal solution.
Whether there exists an jointly differentially private algorithm that both runs in linear time and achieves the optimal tradeoffs between privacy and accuracy characterized in Theorem~\ref{thm:algorithm} and Theorem~\ref{thm:hardness} is an interesting open problem.

Further, the new approach can be implemented in the online random-arrival setting, where the agents show up one by one in a random order and the algorithm must decide the allocation to each agent at his arrival.
It also achieves exact truthfulness if we charge each agent the dual prices in his round because every agent gets his best response bundle.
Previous approach by Hsu et al.~\cite{DBLP:conf/icalp/HsuRRU14, DBLP:conf/soda/Hsu0RW16} gets only approximate truthfulness.

Last but not least, the approach in this section can be implemented in an $\epsilon$-jointly differentially private manner, provided that the supply of each resource is at least $\tilde{O} \big( \frac{m\sqrt{n}}{\alpha\epsilon} \big)$.
Neither the dual multiplicative weight update approach in Theorem~\ref{thm:algorithm} nor the approach in previous work~\cite{DBLP:conf/soda/Hsu0RW16} can achieve $\epsilon$-joint differential privacy, unless the supply $b \ge n$ in which case the problem is trivial.

\begin{theorem}
	\label{thm:online}
	For any packing problem with $m$ constraints, there is a {linear time}, {truthful}, and $(\epsilon, \delta)$-jointly differentially private algorithm whose output is feasible and approximately optimal up to an $\alpha n$ additive factor as long as the supply of each resource is at least $\tilde{O}( \sqrt{mn} / \alpha \epsilon)$.	
	Further, the algorithm can work in the online random-arrival model, and can get $\epsilon$-joint differential privacy if the supply of each resource is at least $\tilde{O}( m \sqrt{n} / \alpha \epsilon)$.
\end{theorem}

We present in Table~\ref{tab:alg-compare} a brief comparison of the dual gradient descent approach by Hsu et al.~\cite{DBLP:conf/soda/Hsu0RW16}, the dual multiplicative weight update approach in Theorem~\ref{thm:algorithm}, and the dual online multiplicative weight update approach in Theorem~\ref{thm:online}.

\begin{table*}
	\renewcommand\arraystretch{1.5}
	\renewcommand{\thefootnote}{\fnsymbol{footnote}}
	\centering
	\begin{tabular}{|c|c|c|c|c|c|c|}
		\hline
		& \multirow{2}{*}{Min supply} & \multirow{2}{*}{\shortstack{Exact\\ feasibility}} & \multirow{2}{*}{$\epsilon$-JDP} & \multirow{2}{*}{Online} & \multirow{2}{*}{\shortstack{Exact\\ truthfulness}} & \multirow{2}{*}{\shortstack{Running\\ time (in $n$)}} \\
		& & & & & & \\
		\hline
		Dual GD~\cite{DBLP:conf/soda/Hsu0RW16} & $\tilde{O} \big( \frac{m}{\alpha \epsilon} \big)$\footnotemark[2] & No & No & No & No & $\tilde{O} (n^3)$ \\	
		\hline
		Dual MWU (\S\ref{sec:algorithm}) & $\tilde{O} \big( \frac{\sqrt{m}}{\alpha \epsilon} \big)$ & Yes & No & No & No & $\tilde{O} (n^3)$ \\
		\hline		
		Dual online MWU (\S\ref{sec:online}) & $\tilde{O} \big( \frac{\sqrt{mn}}{\alpha \epsilon} \big)$ & Yes & Yes\footnotemark[3] & Yes & Yes & $\tilde{O} (n)$ \\
		\hline
		\multicolumn{7}{p{16cm}}{\footnotesize \footnotemark[2]~ Hsu et al.~\cite{DBLP:conf/soda/Hsu0RW16} requires the total supply of all resources to be at least $\tilde{O} \big( \frac{m^2}{\alpha \epsilon} \big)$. We divide it by $m$ and interpret the result as the (average) supply per resource for a direct comparison with the supply requirements in this paper.} \\[-1.5ex]
		\multicolumn{7}{p{16cm}}{\footnotesize \footnotemark[3]~ To get $\epsilon$-JDP, we need a larger supply of at least $\tilde{O} \big(\frac{m\sqrt{n}}{\alpha \epsilon} \big)$ of each resource.}
	\end{tabular}
	
	\caption{A comparison of different approaches}
	\label{tab:alg-compare}
\end{table*}

\bigskip

\textbf{Other related work.}
McSherry and Talwar~\cite{DBLP:conf/focs/McSherryT07} propose the exponential mechanism as a generic method for designing differentially private algorithms for optimization problems for which the set of feasible outcomes do not depend on user data.
As one may expect, such a generic method is not computationally efficient in general.
Bassily et al.~\cite{DBLP:conf/focs/BassilyST14} introduce an efficient implementation of the exponential mechanism when the set of feasible outcomes further forms a compact subset in the Euclidean space and the objective is a convex function.
However, these techniques are not directly applicable to our problem as the set of feasible outcomes of packing problems crucially depends on user data.
Hsu et al.~\cite{DBLP:conf/icalp/HsuRRU14} systematically study what linear programs can be solved in a differentially private manner and what cannot.
However, the packing linear program provably cannot be solved differentially privately~\cite{DBLP:conf/stoc/HsuHRRW14}.

Since the introduction of joint differential privacy by Kearns et al.~\cite{DBLP:conf/innovations/KearnsPRU14}, it has found a wide range of applications, including equilibrium selection~\cite{DBLP:conf/sigecom/RogersR14}, max flow~\cite{DBLP:conf/sigecom/RogersRUW15}, mechanism design~\cite{DBLP:conf/innovations/KearnsPRU14}, privacy-preserving surveys~\cite{ghosh2014buying}, etc.
A common technical ingredient of these work is the billboard lemma introduced by Hsu et al.~\cite{DBLP:conf/stoc/HsuHRRW14}, which also serves as an important building block of the analysis in this paper.

The packing problem has been extensively studied in both the offline (e.g.,~\cite{plotkin1995fast,srinivasan1999improved}) and online settings~(e.g., \cite{buchbinder2009online}).
We note that a lot of these work use the primal dual technique.
Our work can be viewed as an adoption of these techniques in the privacy preserving context.

%% file: preliminaries.tex
\section{Preliminaries}

\subsection{Packing problem and the (partial) dual}

Consider a packing problem with $n$ agents and $m$ resources.
Let $[\ell]$ denote the set $\{1, 2, \dots, \ell\}$ for any positive integer $\ell$.
Each agent $i \in [n]$ demands one of $\ell_i$ bundles of resources.
If we allocate a bundle $k \in [\ell_i]$ to agent $i$, his valuation will be $\pi_{ik} \in [0,1]$ and an $a_{ijk} \in [0,1]$ amount of resource $j$ will be consumed for every $j \in [m]$; if we do not allocate any bundle to agent $i$, his valuation will be $0$ and no resource will be consumed.
The parameters associated with an agent $i$ is the private data of the agent.
Let $U$ denote the data universe and let $D \in U^n$ denote a dataset of $n$ agents.
Each resource $j \in [m]$ has supply $b_j$.
The goal is then to choose a subset of the items that maximizes the total valuation subject to the supply constraints.
This can be formulated as the following packing linear program:
\begin{align*}
\textrm{maximize} ~ & \textstyle \sum_{i\in[n]} \sum_{k \in [\ell_i]} \pi_{ik} x_{ik} \\
\textrm{subject to} ~ & \textstyle \sum_{i\in[n]} \sum_{k \in [\ell_i]} a_{ijk} x_{ik} \le b_j & & \forall j \in [m] \\
& \textstyle \sum_{k \in [\ell_i]} x_{ik} \le 1 & & \forall i \in [n] \\
& x_{ik} \ge 0 & & \forall i \in [n]
\end{align*}
Let $X_i = \{ x_i \in [0, 1]^{\ell_i} : \sum_{k \in [\ell_i]} x_{ik} \le 1\}$ denote the set of feasible decisions associated with agent $i$ and $X = X_1 \times \dots \times X_n$ denote the feasible decisions of all agents if we ignore the supply constraints.
The partial Lagrangian of the above program is:
\begin{equation*}
 \textstyle
\max_{x \in X} \min_{p \in [0, \infty)^m} ~ \left( \sum_{i \in [n]} \sum_{k \in [\ell_i]} \pi_{ik} x_{ik}
- \sum_{j \in [m]} p_j \big( \sum_{i \in [n]} \sum_{k \in [\ell_i]} a_{ijk} x_{ik} - b_j \big)\right)
\end{equation*}

Let $L(x, p)$ denote the above partial Lagrangian objective, that is,
\begin{align*}
L(x, p) & \textstyle = \sum_{i \in [n]} \sum_{k \in [\ell_i]} \pi_{ik} x_{ik}
 - \sum_{j \in [m]} p_j \big( \sum_{i \in [n]} \sum_{k \in [\ell_i]} a_{ijk} x_{ik} - b_j \big) \\
& \textstyle = \sum_{j \in [m]} b_j p_j
 + \sum_{i \in [n]} \sum_{k \in [\ell_i]} x_{ik} \big( \pi_{ik} - \sum_{j \in [m]} a_{ijk} p_j \big)
\end{align*}

Let $D(p) = \max_{x \in X} L(x,p)$ denote the dual objective.
We shall interpret $p_j$ as the unit price of resource $j$ for any $j \in [m]$.
Given a set of prices $p$, an optimal solution of the optimization problem $\max_{x \in X} L(x,p)$ is:
\begin{equation}
x^*_{ik}(p)= \begin{cases} 1, & \mbox{if $\pi_{ik} - \sum_{j=1}^m a_{ijk} p_j \geq 0$ and $k = \argmax_{k' \in [\ell_i]} \pi_{ik'} - \sum_{j=1}^m a_{ijk'} p_j$;} \\
0, &  \mbox{otherwise.} \end{cases}
\end{equation}
Here, we break ties in lexicographical order in the maximization problem of the first case.

The following lemma follows by the Envelope theorem (e.g.,~\cite{milgrom2002envelope}).

\begin{lemma}
\label{lem:envelope}
Given any set of prices $p$, $\nabla_p L\big(x^*(p),p\big)$ is a sub-gradient of $D(p)$.
\end{lemma}
\subsection{Joint differential privacy}

Next, we present necessary preliminaries for differential privacy and joint differential privacy.
Two datasets $D, D^{'} \in U^n$ are $i$-neighbors if they differ only in their $i$-th entry, that is, $D_j=D^{'}_j$ for all $j \neq i$.
The notion of differential privacy by Dwork et al.~\cite{DBLP:conf/tcc/DworkMNS06} requires the output distributions to be similar for any neighboring datasets.

\begin{Definition}[Differential privacy~\cite{DBLP:conf/tcc/DworkMNS06}]
A mechanism $\mathcal{M}: {U}^n \to X$ is $(\epsilon,\delta)$-differentially private if for any $i \in [n]$, any $i$-neighbors $D, D'\in {U}^n$, and any subset $S \subseteq X$, we have:
\[
\textstyle
\Pr \big[\mathcal{M}(D)\in S \big] \leq \exp(\epsilon) \cdot \Pr\big[\mathcal{M}({D}') \in S\big]+\delta ~.
\]
\end{Definition}

The Laplace mechanism by Dwork et al.~\cite{DBLP:conf/tcc/DworkMNS06} computes numerical statistics of a dataset differentially privately by adding Laplacian distributed noise to the output.
We shall use the Laplace mechanism to maintain a sequence of differentially private dual prices.

\begin{Definition}[Laplace mechanism~\cite{DBLP:conf/tcc/DworkMNS06}]
Suppose $f:  {U}^n\rightarrow \mathbb{R}$ has sensitivity $\sigma$, i.e., $|f(D) - f(D')| \le \sigma$ for any neighboring $D$ and $D'$.
Given a database $D\in {U}^n $, the Laplace mechanism outputs $f(D)+Z$, where $Z \sim Lap(\sigma/\epsilon)$.
\end{Definition}

\begin{lemma}[\cite{DBLP:conf/tcc/DworkMNS06}]
\label{lem:laplace}
The Laplace mechanism is $(\epsilon,0)$-differentially private.
\end{lemma}

One can combine differentially private subroutines to obtain algorithms for more complicated tasks; the privacy parameter will scale gracefully.
This is formalized as the composition theorem:

\begin{lemma}[Composition Theorem~\cite{DBLP:conf/focs/DworkRV10}]
\label{lem:composition}
Suppose $\mathcal{A}$ is a $T$-fold adaptive composition of $(\epsilon,\delta)$-differentially private mechanisms.
Then, $\mathcal{A}$ satisfies $(\epsilon', T\delta+\delta')$-differential privacy for
\[
\epsilon'=\epsilon\sqrt{2T\ln(1/\delta')}+T\epsilon(e^{\epsilon}-1) ~.
\]
\end{lemma}

As mentioned in the introduction, for many optimization problems including the packing problem considered in this paper, we need to consider a relaxed notion called joint differential privacy proposed by Kearns et al.~\cite{DBLP:conf/innovations/KearnsPRU14}.
Informally, joint differential privacy is defined w.r.t.\ problems whose outputs are comprised of $n$ components, one for each of the $n$ agents.
It relaxes the requirement of differential privacy so that for any $i$-neighboring datasets, only the output components of agents other than $i$ need to be similarly distributed.

\begin{Definition}[Joint differential privacy~\cite{DBLP:conf/innovations/KearnsPRU14}]
A mechanism $\mathcal{M}: {U}^n \to X = X_1 \times \dots \times X_n$ is $(\epsilon,\delta)$-jointly differentially private if for any $i \in [n]$, any $i$-neighbors $D, D'\in {U}^n$, and any subset $S_{-i} \subseteq X_{-i}$, we have:
\[
\textstyle
\Pr \big[\mathcal{M}_{-i}(D)\in S_{-i} \big] \leq \exp(\epsilon) \cdot \Pr\big[\mathcal{M}_{-i}({D}') \in S_{-i} \big]+\delta ~.
\]
\end{Definition}

The most important connection between joint differential privacy and differential privacy is the following billboard lemma established by Hsu et al.~\cite{DBLP:conf/stoc/HsuHRRW14}.
This lemma has been the cornerstone of many recent work on joint differential privacy and plays a crucial role in this paper as well.

\begin{lemma}[Billboard Lemma~\cite{DBLP:conf/stoc/HsuHRRW14}]
\label{prilem3}
Suppose $\mathcal{M}: {U}^n \rightarrow {Y}^T$ is $(\epsilon,\delta)$-differentially
private. Consider any set of functions $f_i: {U} \times {Y}^T \rightarrow {X}'$. Then the mechanism ${M}^{'}$
 that outputs to
each agent $i: f_i(D_i,{M}(D))$ is $(\epsilon,\delta)$-jointly differentially private.
\end{lemma}

%% file: algorithm.tex
\section{Private dual multiplicative weight algorithm}
\label{sec:algorithm}

Our algorithm is a noisy version of the multiplicative weight update algorithm running in the dual space.
First, recall the partial Lagrangian objective of the packing problem:
\begin{algorithm*}[!htb]
\caption{Private Dual Multiplicative Update Method (Pri-DMW)}
\label{alg:pridmw}
\begin{algorithmic}[1]
\STATE
\textbf{require:}
$b \ge \tilde{O}\big(\tfrac{\sqrt{m}}{\alpha \epsilon}\big)$; we assume $n \ge b$ as the problem is trivial otherwise.
\STATE
\textbf{input:}
objective $\pi$; demands $a_{ijk}$ for all $i\in[n],j\in[m],k\in[\ell_i]$; supply  $b$; approximation parameter $\alpha$; privacy parameter $\epsilon$ and $\delta$.
\STATE
\textbf{parameters:}
initial dual $p^{(1)}=\frac{p_{\max}}{m+1}\cdot \textbf{1}$; number of rounds $T=\frac{\epsilon^2 n^2}{m}$;
step size $\eta=\frac{\ln(m+1)}{\alpha b T}$;
range of dual $p_{\max}=\frac{4 n}{b}$;
privacy parameter in each step $\epsilon'=\frac{\epsilon}{\sqrt{8Tm\ln(2/\delta)}}$;
width $\nabla_{\max} = n + \tfrac{\ln(T)}{\epsilon^{\prime}}$.
%
%
\FORALL{$t = 1, \dots, T$}
\STATE
Let $x^{(t)}_i = x_i^*(p^{(t)})$, i.e., for all $1 \le i \le n$ and $1\le k \le l_i$, let :

$$
x^{(t)}_{ik}= \begin{cases} 1, & \mbox{if $\pi^{(t)}_{ik} - \sum_{j=1}^m a_{ijk} p^{(t)}_j \geq 0$ and   }
 \mbox{ $k = \argmax_{k' \in [\ell_i]} \pi^{(t)}_{ik'} - \sum_{j=1}^m a_{ijk'} p^{(t)}_j$;} \\
0, &  \mbox{otherwise.} \end{cases}
$$
\STATE Sample Laplace noise $\nu_j^{(t)} \sim \textrm{Lap}(\tfrac{1}{\epsilon^{\prime}})$ for $1 \le j \le m$.
\STATE Compute noisy sub-gradient $\nabla_j{\hat{D}}(p^{(t)}) = b - \sum_{i\in[n]}\sum_{k\in[\ell_i]}a_{ijk}x^{(t)}_{ik}+\nu_j^{(t)}$ for $1 \le j \le m$;\\
\STATE Compute truncated noisy sub-gradient $\nabla_j{\bar{D}}(p^{(t)})$ as follows :
$$
\nabla_j{\bar{D}}(p^{(t)}) = \begin{cases}
\nabla_j{\hat{D}}(p^{(t)}) & \mbox{if $- \nabla_{\max} \le \nabla_j{\hat{D}}(p^{(t)}) \le \nabla_{\max}$ ~;} \\
- \nabla_{\max} & \mbox{if $\nabla_j{\hat{D}}(p^{(t)}) < -\nabla_{\max}$~;} \\
\nabla_{\max} & \mbox{if $\nabla_j{\hat{D}}(p^{(t)}) > \nabla_{\max}$~.}
\end{cases}
$$
As for the dummy $(m+1)$-th dimension, let $\nabla_{m+1} {\bar{D}}(p^{(t)}) = 0$.
\STATE Compute $p^{(t+1)}$ such that $p_j^{(t+1)}=\frac{1}{\phi^{(t)}} ~ p_j^{(t)} \big( 1-\eta\nabla_j\bar{D}(p^{(t)}) \big)$ for $1 \le j \le m+1$, where $\phi^{(t)} = \sum_{j = 1}^{m+1} p_j^{(t)} \big( 1-\eta\nabla_j\bar{D}(p^{(t)}) \big)/p_{\max}$ is a normalization term.
\ENDFOR
\STATE \textbf{output:} $\bar{x}_{ik} = \tfrac{1}{T} \sum_{t = 1}^T x^{(t)}_{ik}$ for every $i\in [n],k\in[\ell_i]$, where agent $i$ observes $\bar{x}_i$.
\end{algorithmic}
\end{algorithm*}
\begin{align*}
L(x, p) & \textstyle = \sum_{i \in [n]} \sum_{k \in [\ell_i]} \pi_{ik} x_{ik}
  - \sum_{j \in [m]} p_j \big( \sum_{i \in [n]} \sum_{k \in [\ell_i]} a_{ijk} x_{ik} - b_j \big) \\
&\textstyle= \sum_{i \in [n]} \sum_{k \in [\ell_i]} x_{ik} \big( \pi_{ik} - \sum_{j \in [m]} a_{ijk} p_j \big)
+ \sum_{j \in [m]} b_j p_j
\end{align*}

For simplicity, we assume $b_j = b$ for all $j \in [m]$.
It is straightforward to extend our results to general values of $b_j$'s.
Then, it is without loss to assume that $n \ge b$ as otherwise the optimal solution is trivial with $x_{ik^*} = 1$ where $k^*=\argmax_{k\in[\ell_i]} \pi_{ik}$, for all $i$ .

For convenience of discussion, we add a dummy constraint $\langle 0, x \rangle \le 0$ as the $(m+1)$-th constraint.
As a result, there is a new dual variable $p_{m+1}$ corresponding to the new constraint and, thus, $p$ becomes a $m+1$ dimension vector.
We will restrict $p$ such that its $\ell_1$ norm equals some appropriately chosen $p_{\max}$.

The multiplicative weight update algorithm finds a set of dual prices $p$ that approximately minimizes the dual objective $D(p) = \max_{x \in X} L(x, p)$.
In the process, it also finds an approximately optimal primal solution.
Concretely, it starts with an initial $p^{(1)} = \tfrac{p_{\max}}{m+1} \cdot \mathbf{1}$, where $\mathbf{1}$ is the all-$1$ vector.
In each round $t$, it first computes a sub-gradient of the dual objective $\nabla D(p^{(t)})$ using the envelope theorem, which boils down to computing the best response of the agents to the current dual prices.
Then, it computes $p^{(t+1)}$ by multiplying each entry of $p^{(t)}$ by $1-\eta \nabla_j D(p^{(t)})$ for some appropriately chosen step size $\eta$, and normalizing it to have $\ell_1$ norm $p_{\max}$.
(This is equivalent to a projection back to the simplex $\|p\|_1 = p_{\max}$ with respect to the Kullback-Leibler divergence.)
In order to get joint differential privacy, we use a noisy version of the sub-gradient in our algorithm and show that the error introduced by the Laplacian noise can be bounded.
The algorithm is presented as Algorithm \ref{alg:pridmw}.

\subsection{Proof of Theorem~\ref{thm:algorithm} (privacy)}


The privacy part follows from the next Lemma \ref{lem:privacy} and the Billboard Lemma (Lemma \ref{prilem3}).

\begin{lemma}
\label{lem:privacy}
The sequence of duals $p^{(1)}, \dots ,p^{(T)}$ given by Pri-DMW are $(\epsilon,\delta)$-differentially private.
\end{lemma}

\begin{proof}
Note that the sequence of dual price vectors are determined by a sequence of noisy sub-gradients $\nabla \hat{D}(p^{(t)})$'s.
Hence, it suffices to show the sequence of noisy sub-gradients is $(\epsilon,\delta)$-differentially private.

Next, observe that the noisy sub-gradient $\nabla_j \hat{D}(p^{(t)})$ of each step $t$ is computed by adding Laplace noise of scale $\tfrac{1}{\epsilon'}$ to the sub-gradient $\nabla_j D(p^{(t)}) = b - \sum_{i \in[n]}\sum_{k\in [\ell_i]} a_{ijk} x_{ik}^{(t)}$.
By our assumption, the sub-gradient has sensitivity $1$.
Hence, given $p^{(t)}$, the computation of the noisy sub-gradient $\nabla_j \hat{D}(p^{(t)})$ satisfies $(\epsilon^{\prime}, 0)$-differential privacy (Lemma \ref{lem:laplace}).
Further, $p^{(t)}$ is determined by the noisy sub-gradients $\nabla_j \hat{D}(p^{(1)}), \dots, \nabla_j \hat{D}(p^{(t-1)})$ in previous rounds. Hence, the sequence of noisy sub-gradients are computed via an adaptive composition of $Tm$ Laplacian mechanisms each of which is $(\epsilon^{\prime}, 0)$-differentially private.
The lemma then follows by the composition theorem (Lemma~\ref{lem:composition}).
\end{proof}

\subsection{Proof of Theorem~\ref{thm:algorithm} (approximation)}

In this subsection, we will show that our algorithm violates each constraint by at most $\alpha b$ and is approximately optimal up to an $\alpha n$ additive factor.
Then, to get exact feasibility as in Theorem~\ref{thm:algorithm}, we simply run Pri-DMW with $(1 - \alpha) b$ as the supply per constraint, noting that doing so decreases the optimal objective by at most an $1 - \alpha$ multiplicative factor.

We first introduce some useful facts and technical lemmas.

\begin{lemma}
\label{lem:1}
For any $p$, we have that:
\[
L(x^{(t)},p^{(t)})-L(x^{(t)},p) = \langle p^{(t)} - p, \nabla D(p^{(t)}) \rangle
\]
\end{lemma}
\begin{proof}
By the definition of $L(x, p)$, we have:
\begin{equation*}
L(x^{(t)},p^{(t)})-L(x^{(t)},p)
= \langle p^{(t)} - p, \nabla_p L(x^{(t)},p^{(t)}) \rangle = \langle p^{(t)} - p, \nabla D(p^{(t)}) \rangle ~,
\end{equation*}
where the second equality is due to the envelope theorem (Lemma \ref{lem:envelope}).
\end{proof}

\begin{lemma}
\label{lem:2}
If $b\geq \frac{20\ln(T)\sqrt{m \ln(m+1)\ln(6/\beta) \ln(2/\delta)}}{\alpha \epsilon}$, then we have $\eta \nabla_{\max} < 1$.
\end{lemma}

\begin{proof}
Plug in the value of $\eta, \nabla_{\max},T$.
We have:
\begin{align*}
\begin{split}
\eta \nabla_{\max}
\textstyle
&=\frac{\ln(m+1)}{\alpha bT}(n+\frac{\ln(T)\sqrt{8Tm\ln(2/\beta)}}{\epsilon})\\
&\textstyle
=\frac{\ln(m+1)\sqrt{8m}\ln(T)\ln(2/\delta)}{\alpha \epsilon^2 n b} \leq \frac{\ln(m+1)\sqrt{8m}\ln(T)\ln(2/\delta)}{\alpha \epsilon^2 b^2} ~.
\end{split}
\end{align*}
So if $b \geq \frac{20\ln(T)\sqrt{m \ln(m+1)\ln(6/\beta) \ln(2/\delta)}}{\alpha \epsilon}$, we have $\eta \nabla_{\max} < 1$.
\end{proof}

\begin{lemma}
\label{lem:3}
For any $p$ with $\| p \|_1 = p_{\max}$, and $\eta \nabla_{\max} < 1$, we have:
\begin{equation*}
 D_{KL}(p \| p^{(t+1)})-D_{KL}(p \| p^{(t)})
 \leq -\eta \big\langle p^{(t)}-p,\nabla \bar{D}(p^{(t)}) \big\rangle
 + \eta^2 p_{\max} \nabla_{\max}^2 ~.
\end{equation*}
\end{lemma}

This lemma follows by the standard analysis of multiplicative weight update.
We include the proof in Appendix~\ref{sec:concentration} for the sake of completeness.
The proofs of the next three lemmas are also deferred to Appendix~\ref{sec:concentration}.

\begin{lemma}
\label{lem:4}
For any $q^{(1)}, \dots, q^{(T)}$ such that $\| q^{(t)} \|_1 = p_{\max}$ and that $q^{(t)}$ depends only on $\nu^{(1)}, \dots, \nu^{(t-1)}$, we have that:
\[
\textstyle
\Pr \bigg[ \sum_{t=1}^T \langle q^{(t)}, \nu^{(t)} \rangle \ge \frac{p_{\max}\sqrt{8 T \ln(\frac{6}{\beta})}}{\epsilon^{\prime}} ~ \bigg] \le \tfrac{\beta}{6} ~,
\]
and
\[
\textstyle
\Pr \bigg[ \sum_{t=1}^T \langle q^{(t)}, \nu^{(t)} \rangle \le - \frac{p_{\max}\sqrt{8 T \ln(\frac{6}{\beta})}}{\epsilon^{\prime}} ~ \bigg] \le \tfrac{\beta}{6} ~.
\]
\end{lemma}

\begin{lemma}
\label{lem:5}
For any $q^{(1)}, \dots, q^{(T)} \ge 0$ such that $\| q^{(t)} \|_1 \le  2 p_{\max}$ and $\| q^{(t)} \|_\infty \le  p_{\max}$, and that $q^{(t)}$ depends only on $\nu^{(1)}, \dots, \nu^{(t-1)}$ for $1 \le t \le T$, we have that with probability at most $\beta/3$
\[
\begin{split}
\textstyle
 ~ \sum_{t=1}^T \sum_{j=1}^m q^{(t)}_j \cdot \max \big\{0, \nu^{(t)}_j - \tfrac{\ln(T)}{\epsilon^{\prime}} \big\}
\ge \frac{2p_{\max}\sqrt{8 T \ln(\frac{6}{\beta})}}{\epsilon^{\prime}}
\end{split}
\]
and with probability at most $\beta/3$
\[
\textstyle
 ~ \sum_{t=1}^T \sum_{j=1}^m q^{(t)}_j \cdot \max \big\{0, - \nu^{(t)}_j - \tfrac{\ln(T)}{\epsilon^{\prime}} \big\} \ge  \frac{2p_{\max}\sqrt{8 T \ln(\frac{6}{\beta})}}{\epsilon^{\prime}}  ~
\]
\end{lemma}

Putting the above pieces together, we have the following key lemma that is useful for showing approximate optimality and feasibility.

\begin{lemma}
\label{lem2}
For any $p$ such that $\| p \|_1 = p_{\max}$,
 with probability at least $1 - \beta$, we have:
\[
\textstyle \sum_{t=1}^T \left( L(x^{(t)},p^{(t)}) - L(x^{(t)},p) \right)\leq \frac{1}{\eta} D_{KL}(p \| p^{(1)})
 \textstyle +\eta \cdot T p_{\max} \nabla_{\max}^2+ \frac{20 p_{\max} T \sqrt{m \ln(\frac{6}{\beta}) \ln(\frac{2}{\delta})}}{\epsilon}
\]
\end{lemma}

\subsubsection{Approximate optimality}

\begin{lemma}
\label{lem:6}
For any $t$, we have $L(x^{(t)},p^{(t)}) \ge OPT$.
\end{lemma}

\begin{proof}
Let $x^{*}$ be the optimal primal solution.
%
%
Recall the definition of the Lagrangian objective.
We have that:
\[
 \textstyle
L(x^{(t)},p^{(t)}) =\sum_{i\in[n]}\sum_{k\in[\ell_i]}x^{(t)}_{ik}(\pi_{ik}-\sum_{j\in[m]} a_{ijk}p_j^{(t)})
 +\sum_{j\in[m]} p_j^{(t)}b_{j}~.
\]
Since that $x^{(t)}$ maximizes $L(x,p^{(t)})$, we get that:
\[
 \textstyle
L(x^{(t)},p^{(t)}) \geq \sum_{i\in[n]}\sum_{k\in[\ell_i]}x_{ik}^{*}(\pi_{ik}-\sum_{j \in[m]} a_{ijk}p_j^{(t)}) +\sum_{j\in[m]} p_j^{(t)}b_j~.
\]
Rearranging terms, this is further equal to %
\[
 \textstyle
\sum_{j\in[m]} p^{(t)}_j\left(b_j-\sum_{i\in[n]}\sum_{k\in[\ell_i]} a_{ijk}x_{ik}^{*}\right) +\sum_{i\in[n]}\sum_{k\in[\ell_i]}\pi_{ik}x_{ik}^{*} ~.
\]
Finally, note that by the definition $x^*$, the first term equals $OPT$ and the second term is greater than or equal to zero.
So the lemma follows.
\end{proof}


\begin{proof}{\em (Approximate optimality)}~
Recall that we add a dummy constraint $\langle  0,x \rangle \leq 0$, we let $p_{m+1}=p_{\max}$ and $p_{j}=0$ for all $j\leq m$.
By Lemma \ref{lem2}, the following holds with probability at least $1-\beta$:
\begin{align*}
\textstyle
\sum_{t=1}^T \left( L(x^{(t)},p^{(t)}) - L(x^{(t)},p) \right)
& \leq
\frac{1}{\eta} D_{KL}(p \| p^{(1)}) +\eta \cdot T p_{\max} \nabla_{\max}^2  + \frac{20 p_{\max} T \sqrt{m \ln(6/\beta) \ln(2/\delta)}}{\epsilon}\\
&
\textstyle
=\frac{1}{\eta}p_{\max}\ln(m+1)+\eta\cdot Tp_{\max}(n+\frac{\ln(T)}{\epsilon^{'}})^2
 +\frac{20 p_{\max} T \sqrt{m \ln(6/\beta) \ln(2/\delta)}}{\epsilon}\\
&
\textstyle
=4\alpha Tn+\frac{T\ln(T)^2\ln(2/\delta)\ln(m+1)mn}{\alpha \epsilon^2b^2} +\frac{20nT \sqrt{m \ln(6/\beta) \ln(2/\delta))}}{b\epsilon}.
\end{align*}
If $b\geq \frac{20\ln(T)\sqrt{m \ln(m+1)\ln(6/\beta) \ln(2/\delta)}}{\alpha \epsilon}$, we further bound the 2nd and 3rd terms by
\[
\textstyle
\frac{T\ln(T)^2\ln(2/\delta)\ln(m+1)mn}{\alpha \epsilon^2b^2}\leq \alpha T n ~,
\text{and }\quad  \frac{20nT \sqrt{m \ln(6/\beta) \ln(2/\delta)}}{b\epsilon}\leq \alpha Tn ~.
\]
So we get that
\[
\textstyle
\sum_{t=1}^T \left( L(x^{(t)},p^{(t)}) - L(x^{(t)},p)\right)\leq 6\alpha Tn ~.
\]
Note that $\sum_{i\in[n]}\sum_{k\in[\ell_i]}\pi_{ik} \bar{x}_{ik}=ALG$.
By our choice of $p$, we have $\sum_{t=1}^TL(x^{(t)},p)=T\cdot ALG$.
By Lemma \ref{lem:6}, we have $\sum_{t=1}^TL(x^{(t)},p^{(t)}) \ge T \cdot OPT$.
So we have:
\[
\textstyle
T(OPT-ALG) \leq \sum_{t=1}^T\left(L(x^{(t)},p^{(t)})-L(x^{(t)},p)\right)\leq 6T \alpha n ~,
\]
So we have the desired approximate optimality guarantee.
\end{proof}

\subsubsection{Approximate feasibility}
\begin{proof}{\em (Approximate feasibility)}~
We choose $p$ to penalize the over-demands and, thus, make $L(x^{(t)},p)$ as small as possible.
We let $p_{j^{*}}=p_{\max}$, where $j^{*}$ is the most over-demanded constraint and let $p_{j}=0$ for any $j\neq j^{*}$.
Let $s = b - \sum_{i\in[n]}\sum_{k\in[\ell_i]} a_{ij^*k} \bar{x}_{ik}$ be the over-demand of $j^*$.
By Lemma \ref{lem2} and the choice of $p$ and $p^{(1)}$, with probability at least $1-\beta$, we have:
\begin{align}
\begin{split}
\label{eq7}
&\textstyle
\sum_{t=1}^T \left( L(x^{(t)},p^{(t)}) - L(x^{(t)},p) \right) \\
&
\textstyle \qquad \qquad
\leq
\frac{1}{\eta} D_{KL}(p \| p^{(1)}) +\eta \cdot T p_{\max} \nabla_{\max}^2  + \frac{20 p_{\max} T \sqrt{m \ln(6/\beta) \ln(2/\delta)}}{\epsilon} \\
&
\textstyle \qquad \qquad
=
\frac{1}{\eta} p_{\max} \ln(m+1) +\eta \cdot T p_{\max} \nabla_{\max}^2  + \frac{20 p_{\max} T \sqrt{m \ln(6/\beta) \ln(2/\delta)}}{\epsilon}
~.
\end{split}
\end{align}
By the choice of $p$, we further get that:
%
\[
\textstyle
\sum_{t=1}^TL(x^{(t)},p)
\textstyle
=T\cdot ALG +p_{\max}\sum_{t=1}^T( b-\sum_{i\in[n]}\sum_{k\in[\ell_i]}a_{ij^{*}k}x_{ik}^{(t)}) =T(ALG-p_{\max}s) ~.
\]
Putting together with Lemma \ref{lem:6}, the LHS of \eqref{eq7} is at least $T(OPT-ALG+p_{\max}s)$.

Also note that $ALG \leq (1+\frac{s}{b})OPT$, because increasing the supply per resource from $b$ to $b+s$ increases the optimal packing objective by at most a $\frac{b+s}{b}$ factor.
So we have:
\begin{align*}
\mbox{ LHS of } (\ref{eq7}) &
\textstyle
\geq T\left(OPT-(1+\frac{s}{b})OPT+p_{\max} s\right) \\
&
\textstyle
=Ts\left(p_{\max}-\frac{1}{b} OPT\right)  \geq Ts(p_{\max}-\frac{n}{b})=\frac{3Tsp_{\max}}{4} ~.
\end{align*}
So we have that:
\[
\frac{3Tsp_{\max}}{4}
\leq  \textstyle \frac{1}{\eta} p_{\max} \ln(m+1) +\eta \cdot T p_{\max} \nabla_{\max}^2 + \frac{20 p_{\max} T \sqrt{m \ln(6/\beta) \ln(2/\delta)}}{\epsilon} ~.
\]
Plug in the choice of parameters, we get that
\[
\textstyle
\frac{s}{b} \le O\big( \alpha+\frac{\ln(T)^2\ln(m+1)\ln(2/\delta)m}{\alpha \epsilon^2 b^2}+\frac{\sqrt{m \ln(6/\beta) \ln(2/\delta)}}{b\epsilon} \big) ~.
\]
Therefore, the max violation per constraint is at most $s \leq O(\alpha b)$ as long as the supply per constraint is at least $b\geq \frac{20\ln(T)\sqrt{m \ln(m+1)\ln(6/\beta) \ln(2/\delta)}}{\alpha \epsilon}$.
\end{proof}

%% file: hardness.tex
\section{Hardness (Theorem~\ref{thm:hardness})}
\label{sec:hardness}

\begin{theorem}[Theorem 1.1 of \cite{steinke2017between}]
	\label{thm:hardness-approximate-dp}
	Suppose there is an $(\epsilon, \delta)$-differentially private algorithm for answering $m$ arbitrary counting queries on a dataset of size $b$ with average error at most $O(\alpha b)$.
	Then, we have
	\[
	\textstyle
	b \ge \Omega \left( \frac{\sqrt{m \log(1/\delta)}}{\epsilon\alpha} \right) ~.
	\]
\end{theorem}

\begin{lemma}
	\label{lem:hardness-reduction}
	Suppose for some $b$ and $m$, there is an $(\epsilon, \delta)$-jointly differentially private algorithm that with high probability outputs a feasible solution that is optimal up to an $\alpha n$ additive factor for $n = \Theta(b)$.
	Then, there is an $(\epsilon, \delta)$-differentially private algorithm for answering $m$ arbitrary counting queries on any dataset of size $\Theta(b)$ with average error at most $O(\alpha b)$.
\end{lemma}

\begin{proof}
	Consider an arbitrary dataset of size $n' = \frac{b}{2}$, denoted as $D' = \{ d_1, \dots, d_{n'}\}$, and an arbitrary set of $m$ counting queries of sensitivity $1$, denoted as  $\mathcal{Q} = \{q_1, \dots, q_m\}$.
	Construct an instance of the packing problem with $n = \Theta(b)$ agents as follows:
	
	Let there be a set $n'$ agents, denoted as $A$, each of which demands a unique bundle, i.e., $|[\ell_i]| = 1$ and therefore we omit subscript $k$ in the following.
	The resource demanded in the bundle is $a_{ij} = q_j(d_i)$ and the value is $\pi_i = 1$.
	
	Further, let there be $2b$ agents, denoted as $B$, each of whom demands any subset of size $\frac{m}{2}$ and has value $\frac{1}{4}$.
	That is, $\ell_i = {m \choose {m}/{2}}$; for any subset $S \subseteq [m]$ of size $\frac{m}{2}$, let there be a $k$ such that $a_{ijk} = 1$ if $j \in S$ and $0$ otherwise, and $\pi_{ik} = \frac{1}{4}$.
	
	Note that by allocating a bundle to one of the agents in $A$, we get value at least $\frac{1}{m}$ per unit of resources.
	On the other hand, allocating a bundle to one of the agents in $B$ gets at only $\frac{1}{2m}$ value per unit of resources.
	
	\begin{lemma}
		\label{lem:hardness-opt-bound}
		$OPT \ge n' + \frac{1}{2m} \sum_{j \in [m]} \big( b - q_j(D') \big) - \frac{1}{2}$.
	\end{lemma}

	\begin{proof}
		We will prove by constructing a feasible solution with total value lower bounded by the RHS of the inequality.
		First, we will allocate to all agents in $A$ their desired bundles.
		We gain $n'$ total value by doing so, and has a remaining supply $b - q_j(D')$ of resource $j$ for any $j \in [m]$.
		
		Then, we claim that it is possible to allocate bundles to a subset of the agents in $B$ such that we use up all but at most $1$ unit of every resource.
		Given that, the lemma follows because allocating bundles to agents in $B$ gives precisely $\frac{1}{2m}$ value per unit of resources.
		
		In the rest of the proof, we will explain how to allocate bundles to agents in $B$.
		Note that after allocating bundles to agents in $A$, the maximum demand of any resource is at most $n' = \frac{b}{2}$, while the minimum demand of some resource could be $0$.
		
		We will inductively decide how to allocate bundles to agents in $B$ in $\frac{b}{2} - 1$ rounds such that after round $i$, $0 \le i \le \frac{b}{2} - 1$, the maximum demand of any resource will be at most $n' + i$, while the minimum demand of any resource will be at least $2i$.
		Then, at the end of the process, the maximum demand will be at most $b - 1$ and the minimum demand will be at least $b - 2$.
		We simply allocate to two more agents such that the first one gets resources $1$ to $\frac{m}{2}$ and the second one gets resources $\frac{m}{2} + 1$ to $m$.
		That is, we further use one unit of each resource.
		
		The claim is vacuously true for $i = 0$.
		Next, suppose we have finished the first $i - 1$ rounds for some $i \ge 1$.
		Let us explain how to allocate bundles in round $i$.
		
		Suppose the maximum and minimum demands of resources differs by at most $1$.
		We simply repeatedly allocate bundles to two agents in set $B$ so that one unit of each resource is allocated to exactly one of the two agents until the maximum demand equals $n' + i$.
		
		Otherwise, suppose the maximum and minimum demands, denoted as $d^+$ and $d^-$ respectively, differ by at least $2$.
		We will further divide it into three cases depending on the numbers of resources with demands $d^+$ and $d^-$ respectively, denoted as $k^+$ and $k^-$ respectively.
		Let us assume w.l.o.g.\ that the resources are sorted in ascending order of their current demands.
		E.g., resources $1$ to $k^-$ are those with demands equal $d^-$, and resource $m - k^+ + 1$ to $m$ are those with demands equal $d^+$.
		
		The first case is when there are $k^- \le \frac{m}{2}$ resources with demand $d^-$ and $k^+ > k^-$.
		In this case, consider two agents in $B$.
		Let us allocate items $1$ to $\frac{m}{2}$ to the first agent, and allocate items $1$ to $k^-$ together with items $\frac{m}{2} + 1$ to $m - k^-$ to the second one.
		Note that $m - k^- > m - k^+$ by our assumption on $k^-$ and $k^+$.
		We have increased (1) the demands of resources $1$ to $k^-$ by $2$ (i.e., from $k^-$ to $k^- + 2 \le k^+$), (2) the demands of resources $k^- + 1$ to $m - k^+$ by $1$ (i.e., from $< k^+$ to at most $k^+$), and (3) demands of a subset of the resources $m - k^+ + 1$ to $m$ by $1$ (i.e., from $k^+$ to $k^+ + 1$).
		Thus, we have achieved the desired goal in round $i$.

		The second case is when $k^- \le \frac{m}{2}$ and $k^+ \le k^-$.
		We consider the same two agents as in the previous case.
		After allocating to those two agents, we have increased (1) the demands of resources $1$ to $k^-$ by $2$ (i.e., from $k^-$ to $k^- + 2 \le k^+$), and (2) the demands of a subset of the resources $k^- + 1$ to $m - k^+$ by $1$ (i.e., from $< k^+$ to at most $k^+$).
		Then, we further allocate to two agents in set $B$ so that one unit of each resource is allocated to exactly one of the two agents.
		Then, we have increased (1) the demands of resources $1$ to $k^-$ by $3$ (i.e., from $k^-$ to at most $k^- + 3 \le k^+ + 1$), (2) the demands of resources $k^- + 1$ to $m - k^+$ by either $1$ or $2$ (i.e., from $< k^+$ to at most $k^+ + 1$), and (3) the demands of resources $m - k^+ + 1$ to $m$ by $1$ (i.e., from $k^+$ to at most $k^+ + 1$).
		Thus, we have achieved the desired goal in round $i$.
				
		The final case is when $k^- > \frac{m}{2}$.
		In this case, consider allocating to $4$ agents.
		The first and second agents get resources $1$ to $\frac{m}{2}$;
		the third agent gets resources $\frac{m}{2} + 1$ to $m$;
		and the fourth agent gets resources $k^- - \frac{m}{2} + 1$ to $k^-$.
		Then, we have increased (1) the demands of resources $1$ to $k^-$ by either $2$ or $3$ (i.e., from $k^-$ to at most $k^- + 3 \le k^+ + 1$), (2) the demands of all other resources by $1$ (i.e., from $\le k^+$ to at most $k^+ + 1$).
		Thus, we have achieved the desired goal in round $i$.
	\end{proof}
	
	\begin{lemma}
		\label{lem:hardness-leftout-agents}
		Any solution that is optimal up to an $\alpha n$ additive factor must allocate bundles to all but at most $2 \alpha n + 1$ agents in $A$.
	\end{lemma}
	
	\begin{proof}
		We will prove the lemma even for fractional solutions.
		As the solution is allowed to be fractional and there are plenty of agents in $B$, we can assume without loss that all resources are fully allocated.
		If we allocate to all agents in $A$, the objective would be $n' + \frac{1}{2m} \sum_{j \in [m]} \big( b - q_j(D') \big)$.
		Note that the value per unit of resource of allocating to agents in set $B$ is at most a half of that of allocating to agents in set $A$.
		Thus, for each agent $i \in A$ that remains unallocated in the solution, the objective decreases by at least $\frac{1}{2}$ even if we fully allocate the resources that were allocated to the $i$ to some other agents in $B$.
		Putting together with Lemma~\ref{lem:hardness-opt-bound} proves the lemma.
	\end{proof}

	\begin{lemma}
		\label{lem:hardness-unused-resources}
		Any solution that is optimal up to an $\alpha n$ additive factor must allocate all but at most $O(\alpha b m)$ units of the resources.
	\end{lemma}

	\begin{proof}
		Again, we will prove the lemma even for fractional solutions.
		If all resources are fully allocated and we optimally allocate to all agents in $A$, the objective is $n' + \frac{1}{2m} \sum_{j \in [m]} \big( b - q_j(D') \big)$.
		Since the value per unit of resources is at least $\frac{1}{2m}$ for any agent, putting together with Lemma~\ref{lem:hardness-opt-bound} proves the lemma.
	\end{proof}

	Now we are ready to introduce the reduction from differentially private query release to jointly differentially private packing.
	By solving the constructed packing instance in an $(\epsilon, \delta)$-jointly differentially private manner, the allocation for agents in $B$ is $(\epsilon, \delta)$-differentially private w.r.t.\ the data of agents in $A$ according to the definition of joint differential privacy.
	Then, we can output $b$ minus the number of units of resource $j$ allocated to agents in $B$, denoted as $\tilde{q}_j$, as the response for query $q_j$.
	Since this is a post-processing on the output of an $(\epsilon, \delta)$-differentially private algorithm, the responses are $(\epsilon, \delta)$-differentially private as well.
	
	It remains to analyze the accuracy of the responses.	
	On one hand, $\tilde{q}_j$ is greater than or equal to the number of units of resource $j$ allocated to agents in $A$, which, by Lemma~\ref{lem:hardness-leftout-agents} is at least $q_j(D') - O(\alpha n) = q_j(D') - O(\alpha b)$.
	On the other hand, $\tilde{q}_j$ is at most the number of units of resource $j$ allocated to agents in $A$ plus the number of unallocated units of resource $j$.
	The former is at most $q_j(D')$ while the latter is at most $O(\alpha b)$ on average according to Lemma~\ref{lem:hardness-unused-resources}.
	Putting together $\tilde{q}_j$'s have average error at most $\alpha$.
\end{proof}

%% file: online.tex
\section{Private dual online multiplicative weight algorithm}
\label{sec:online}

\begin{algorithm*}[!htb]
\caption{Private Dual Online Multiplicative Update Method (Pri-DOMW)}
\label{alg:online-pridmw}
\begin{algorithmic}[1]
\STATE
\textbf{input:}
objective $\pi$; demands $a_{ijk}$ for all $i\in[n],j\in[m],k\in[\ell_i]$; supply $b$; approximation parameter $\alpha$; privacy parameter $\epsilon$ and $\delta$.
\STATE
\textbf{parameters:}
initial dual $p^{(1)}=\frac{p_{\max}}{m+1}\cdot \textbf{1}$;
step size $\eta = \frac{1}{\sqrt{n}\sigma}$;
range of dual $p_{\max} = \frac{\alpha n}{\sigma}$;
noise scale $\sigma = \frac{m}{\epsilon}$ (if $\delta = 0$) or $\frac{1}{\epsilon}{\sqrt{8 m \ln (1/\delta)}}$ (if $\delta > 0$) in each step;
width $\nabla_{\max} = 1 + \sigma \cdot \ln n$.
\STATE
\textbf{require:}
$b \ge \tilde{O}\big(\tfrac{\sqrt{n} \sigma}{\alpha}\big)$; we assume $n \ge b$ as the problem is trivial otherwise.
\STATE
Pick a permutation $\lambda$ over $[n]$ uniformly at random.
\FORALL{$t = 1, \dots, n$}
\STATE
Let $x_{\lambda(t)} = x^*_{\lambda(t)} (p^{(t)})$, i.e., for all $k \in [\ell_{\lambda(t)}]$, let
%
$$
\textstyle
x_{\lambda(t)k} =
\begin{cases}
1, & \mbox{if $\pi_{\lambda(t)k}-\sum_{j=1}^ma_{\lambda(t)jk}p^{(t)}_j \geq 0, $} \mbox{$k=\argmax_{k'\in [\ell_i]}{\pi_{\lambda(t)k'} - \sum_{j=1}^ma_{\lambda(t)jk'}p^{(t)}_j}$ ;}\\
0, & \mbox{otherwise.}
\end{cases}
$$
\STATE Sample Laplacian noise $\nu_{tj} \sim \textrm{Lap}\big(\sigma\big)$ for all $j \in [m]$.
\STATE Let $y_{t} = \sum_{k \in [\ell_{\lambda(t)}]} x_{\lambda(t)k} \cdot (a_{\lambda(t)1k}, \dots, a_{\lambda(t)mk})$ be the demand vector of agent $t$.
\STATE Let $z_{t} = y_{t} + \nu_{t}$ be the noisy demand vector.
\STATE Compute $p^{(t+1)}$ such that $p_j^{(t+1)} \propto p_j^{(t)} \big( 1 + \eta \cdot (z_{tj} - \frac{b}{n}) \big)$ for $j \in [m]$ and $p_{m+1}^{(t+1)} \propto 1$.
\ENDFOR
\STATE \textbf{output:} $x_{\lambda(1)},\dots,x_{\lambda(n)}$.
\end{algorithmic}
\end{algorithm*}

In this section we introduce an alternative algorithm for solving the packing problem in a jointly differentially private manner.
This alternative approach is similar to the previous one, with the following differences.
In each step, instead of computing the best responses of all agents for the current dual prices and, thus, compute the corresponding subgradient, we simply pick one of the agents and use his best response to compute a proxy subgradient.
The agent then gets the bundle specified by his best response.
We will choose a random ordering of the agents at the beginning and pick agents in that order.
As a result, the algorithm will update dual prices for only $n$ rounds as oppose to $\frac{\epsilon^2 n^2}{m}$ rounds in the previous approach.

There are both pros and cons of this alternative approach.

\begin{itemize}
	\item The main disadvantage is that it requires a much larger supply to get the same approximation guarantees.
	Intuitively, this is because (1) we use proxy subgradients in place of the actual subgradients and, thus, introduce some extra error, and (2) it goes over the each agent only once and, thus, does not optimize the number of rounds of dual updates.
	\item For the same two reasons that cause the above drawback, the approach in this section has the advantage that we get incentive compatibility for free if we charge the agent the corresponding dual prices in his round, since every agent gets the best response bundle.
	\item Further, it can be implemented in the online random-arrival setting, where the agents show up one by one in a random order and the algorithm must decide the allocation to each agent at his arrival.
	\item Last but not least, the approach in this section can be implemented in an $\epsilon$-jointly differentially private manner.
	Neither the dual multiplicative weight update approach in Section~\ref{sec:algorithm} nor the approach in previous work~\cite{DBLP:conf/soda/Hsu0RW16} can achieve $\epsilon$-joint differential privacy.
\end{itemize}

\subsection{Proof of Theorem~\ref{thm:online} (privacy)}

By standard privacy properties of the Laplace mechanism and the composition theorem, the noisy demand vectors $z_t$'s are $\epsilon$-differentially private if the noise scale is $\sigma = \frac{m}{\epsilon}$, and are $(\epsilon, \delta)$-differentially private if the noise scale is $\sigma = \frac{1}{\epsilon} \sqrt{8m \ln({1}/{\delta})}$.
Then, the joint differential privacy of Algorithm~\ref{alg:online-pridmw} follows by the Billboard Lemma (Lemma~\ref{prilem3}).

\subsection{Proof of Theorem~\ref{thm:online} (approximation)}

\subsubsection{Key lemma}

We will first establish a key lemma that is an analogue of Lemma~\ref{lem2} in the previous section.

\begin{lemma}
	\label{lem:online-key}
	For any given $p$ such that $\| p \|_1 = p_{\max}$, with high probability, we have:
	\[
	\textstyle
	OPT - L(x,p)
	 \leq
	\frac{1}{\eta} D_{KL}(p \| p^{(1)}) + \eta \cdot \tilde{O} \big( n p_{\max} \sigma^2 \big)  + \tilde{O} \big(\sqrt{n} p_{\max} \sigma \big) ~.	
\]
\end{lemma}

The proof of Lemma~\ref{lem:online-key} follows by a sequence of technical lemmas as follows.

\begin{lemma}
	\label{lem:online-one-step-1}
	For any $p$ with $\| p \|_1 = p_{\max}$, and $\eta \nabla_{\max} < 1$, we have:
	\begin{equation*}
	 D_{KL}(p \| p^{(t+1)})-D_{KL}(p \| p^{(t)}) \leq \eta \big\langle p^{(t)}-p, z_{t} - \tfrac{b}{n} \cdot \mathbf{1} \big\rangle + \eta^2 p_{\max} \nabla_{\max}^2 ~.
	\end{equation*}
\end{lemma}

We will omit the proof of the above lemma because it is essentially the same as that of Lemma~\ref{lem:3}, replacing $- \nabla \bar{D}(p^{(t)})$ with $z_t - \frac{b}{n} \cdot \mathbf{1}$.

Next, we decompose the Lagrangian objective $L(x, p)$ into the sum of $n$ components $L_i(x_i, p)$'s, $i \in [n]$, as follows:
\[
\textstyle
L_i(x_i, p) = \sum_{k \in [\ell_i]} \pi_{ik} x_{ik} - \sum_{j \in [m]} p_j \big( \sum_{k \in [\ell_i]} a_{ijk} x_{ik} - \frac{b}{n} \big) ~.
\]
Then, we have:
\begin{equation}
\label{eqn:online-one-step}
\begin{split}
 \big\langle p^{(t)}-p, y_{t} - \tfrac{b}{n} \cdot \mathbf{1} \big\rangle
= L_{\lambda(t)}(x_{\lambda(t)}, p) - L_{\lambda(t)}(x_{\lambda(t)}, p^{(t)}) ~.
\end{split}
\end{equation}

\begin{lemma}
	\label{lem:online-one-step-2}
	For any $p$ with $\| p \|_1 = p_{\max}$, and $\eta \nabla_{\max} < 1$, any fixed permutation $\lambda$ and over the randomness of the Laplacian noise, we have that with high probability:
	\begin{equation*}
\begin{split}
	& \textstyle
	\sum_{t \in [n]} L_{\lambda(t)}(x_{\lambda(t)}, p^{(t)}) - L(x, p) \\
& \textstyle\quad \le \frac{1}{\eta} D_{KL}(p \| p^{(1)}) + \eta \cdot \tilde{O} \big( n p_{\max} \sigma^2 \big)
\tilde{O} \big(\sqrt{n} p_{\max} \sigma \big) ~,
\end{split}
	\end{equation*}
\end{lemma}

\begin{proof}
	Recall that $z_t = y_t + \nu_t$.
	By Eqn.~\eqref{eqn:online-one-step} and Lemma~\ref{lem:online-one-step-1}, we have:
	\begin{align*}
	&
	D_{KL}(p \| p^{(t+1)})-D_{KL}(p \| p^{(t)}) \\
    & \qquad \textstyle \leq \eta \cdot \big( L_{\lambda(t)}(x_{\lambda(t)}, p) - L_{\lambda(t)}(x_{\lambda(t)}, p^{(t)}) \big) + \eta^2 p_{\max} \nabla_{\max}^2 + \eta \cdot \langle p^{(t)} - p, \nu_t \rangle ~.
	\end{align*}
	Summing over $t \in [n]$, we get that
	\begin{align*}
	& \textstyle D_{KL}(p \| p^{(n+1)})-D_{KL}(p \| p^{(1)}) \\
	& \qquad  \textstyle  \leq \eta \cdot \sum_{t \in [n]}\big( L_{\lambda(t)}(x_{\lambda(t)}, p) - L_{\lambda(t)}(x_{\lambda(t)}, p^{(t)}) \big) + \eta^2 n p_{\max} \nabla_{\max}^2+ \eta \cdot \sum_{t \in [n]} \langle p^{(t)} - p, \nu_t \rangle ~.
	\end{align*}
	
	Since $\lambda$ is a permutation, we have that $\sum_{t \in [n]} L_{\lambda(t)} (x_{\lambda(t)}, p) = \sum_{t \in [n]} L_t (x_t, p) = L(x, p)$.
	Further, recall that our choice of $\nabla_{\max} = \tilde{O}(\sigma)$.
	So the RHS further equals
	\[
\begin{split}
	& \textstyle
	\eta \cdot \big( L(x, p) - \sum_{t \in [n]} L_{\lambda(t)}(x_{\lambda(t)}, p^{(t)}) \big) + \eta^2 \cdot \tilde{O} \big( n p_{\max} \sigma^2 \big) + \eta \cdot \sum_{t \in [n]} \langle p^{(t)} - p, \nu_t \rangle ~.
\end{split}
	\]
	
	Finally, let us consider the randomness of the Laplacian noise.
	We will use the concentration bound for martingales \cite{azuma1967weighted} to bound the last term.
	By Lemma~\ref{lem:4}, the last term on the RHS is at most $\tilde{O}(\sqrt{n} p_{\max} \sigma)$ with high probability.
	Rearrange terms and the lemma follows.
\end{proof}

\begin{lemma}
	\label{lem:online-lagrangian-bound}
	With high probability, we have:
	\[
	\textstyle
	\sum_{t \in [n]} L_{\lambda(t)}(x_{\lambda(t)}, p^{(t)}) \ge OPT - O(\sqrt{n} p_{\max}) ~.
	\]
\end{lemma}

\begin{proof}
	We will proceed in two steps.
	Firstly we will show that the expectation of the LHS is at least $OPT - \tilde{O}(\sqrt{n} p_{\max})$.
	Then, we will use the standard concentration bound for martingales to bound the deviation of the LHS from its expectation.
	
	For any $t$, let us fix the randomness in the first $t-1$ rounds and, thus, fix $p^{(t)}$.
	Taking expectation over only the randomness of round $t$, we get that:
	\[
	\textstyle
	\E \big[ L_{\lambda(t)}(x_{\lambda(t)}, p^{(t)}) \,|\, \lambda(1:t-1) \big]  = \frac{1}{n - t + 1} \sum_{i \in [n] \setminus \lambda(1:t-1)} L_i(x_i^*(p^{(t)}), p^{(t)}) ~.
	\]
	
	Let $x^*$ be the offline optimal primal solution.
	Since $x^*_i(p^{(t)})$ is the best response to $p^{(t)}$, we have:
	\[
\begin{split}
	 \textstyle
	\E \big[ L_{\lambda(t)}(x_{\lambda(t)}, p^{(t)}) \,|\, \lambda(1:t-1) \big]  \ge \frac{1}{n - t + 1} \sum_{i \in [n] \setminus \lambda(1:t-1)} L_i(x_i^*, p^{(t)}) ~.
\end{split}
	\]
	
	Next, consider the difference between the above quantity and the actual average over $n$ agents, i.e., $\frac{1}{n} \sum_{i \in [n]} L_i(x_i^*, p^{(t)})$.
	We have:
	\begin{align*}
	&
	\textstyle
	\frac{1}{n - t + 1} \sum_{i \in [n] \setminus \lambda(1:t-1)} L_i \big( x_i^*, p^{(t)} \big) - \frac{1}{n} \sum_{i \in [n]} L_i \big( x_i^*, p^{(t)} \big) \\
	&
	\textstyle \qquad
	= \left( \frac{1}{n - t + 1} \sum_{i \in [n] \setminus \lambda(1:t-1)} \sum_{k \in [\ell_j]} \pi_{ik} x^*_{ik}   - \frac{1}{n} \sum_{i \in [n]} \sum_{k \in [\ell_j]} \pi_{ik} x^*_{ik} \right) \\
	&
	\textstyle
	\qquad \quad + \sum_{j \in [m]} p_j \cdot \left( \frac{1}{n-t+1} \sum_{i \in [n] \setminus \lambda(1:t-1)} \sum_{k \in [\ell_i]} a_{ijk} x^*_{ik}  - \frac{1}{n} \sum_{i \in [n]} \sum_{k \in [\ell_i]} a_{ijk} x^*_{ik} \right) \\
	&
	 \textstyle \qquad
	\ge \large\left( \frac{1}{n - t + 1} \sum_{i \in [n] \setminus \lambda(1:t-1)} \sum_{k \in [\ell_j]} \pi_{ik} x^*_{ik}  - \frac{1}{n} \sum_{i \in [n]} \sum_{k \in [\ell_j]} \pi_{ik} x^*_{ik} \right) \\
	&
	\textstyle
	\qquad \quad - p_{\max} \cdot  \max_{j \in [m]} \left| \frac{1}{n-t+1} \sum_{i \in [n] \setminus \lambda(1:t-1)} \sum_{k \in [\ell_i]} a_{ijk} x^*_{ik}  - \frac{1}{n} \sum_{i \in [n]} \sum_{k \in [\ell_i]} a_{ijk} x^*_{ik} \right| ~.
	\end{align*}

	Note that $[n] \setminus \lambda(1:t-1)$ is a random subset of $n-t+1$ elements in $[n]$.
	By the standard concentration bound for sampling without replacement \cite{DBLP:conf/soda/AgrawalD15}, with high probability over the randomness of $\lambda(1:t-1)$, we have that
	\[
	 \textstyle
	\frac{1}{n - t + 1} \sum_{i \in [n] \setminus \lambda(1:t-1)} \sum_{k \in [\ell_j]} \pi_{ik} x^*_{ik} - \frac{1}{n} \sum_{i \in [n]} \sum_{k \in [\ell_j]} \pi_{ik} x^*_{ik} \ge - \tilde{O} \big( \frac{1}{\sqrt{n-t+1}} \big) ~,
	\]
	and for any $j \in [m]$
	\[
	 \textstyle
	\left| \frac{1}{n-t+1} \sum_{i \in [n] \setminus \lambda(1:t-1)} \sum_{k \in [\ell_i]} a_{ijk} x^*_{ik}- \frac{1}{n} \sum_{i \in [n]} \sum_{k \in [\ell_i]} a_{ijk} x^*_{ik} \right| \le \tilde{O} \big( \frac{1}{\sqrt{n-t+1}} \big) ~.
	\]

	Putting together, we have:
	\begin{align*}
	\textstyle
	 \E \big[ L_{\lambda(t)}(x_{\lambda(t)}, p^{(t)}) \,|\, \lambda(1:t-1) \big]
	&
	\textstyle
	\ge \frac{1}{n - t + 1} \sum_{i \in [n] \setminus \lambda(1:t-1)} L_i \big( x_i^*, p^{(t)} \big) \\
	&
	\textstyle
	\ge \frac{1}{n} \sum_{i \in [n]} L_i \big( x_i^*, p^{(t)} \big) - \tilde{O} \big( \frac{p_{\max}}{\sqrt{n-t+1}} \big) \\
	&
	\textstyle
	\ge
	\frac{1}{n} OPT - \tilde{O} \big( \frac{p_{\max}}{\sqrt{n-t+1}} \big) ~,
	\end{align*}
	where the last inequality follows by the optimality of $x^*$.
	
%
%
	Summing over $t \in [n]$, noting that $\sum_{t \in [n]} \frac{1}{\sqrt{n - t + 1}} = \sum_{t \in [n]} \frac{1}{\sqrt{t}} = O(\sqrt{n})$, we have:
	\begin{equation}
	\label{eqn:online-lagrangian-bound}
	 \textstyle
	\E \big[ L_{\lambda(t)}(x_{\lambda(t)}, p^{(t)}) \,|\, \lambda(1:t-1) \big]
	\ge OPT - \tilde{O}(\sqrt{n} p_{\max}) ~.
	\end{equation}
	
	Then, consider a sequence of random variables $u_t \in [-2p_{\max}, 2p_{\max}]$'s as follows:
	\[
	\textstyle
	u_t = L_{\lambda(t)}(x_{\lambda(t)}, p^{(t)}) - \E \big[ L_{\lambda(t)}(x_{\lambda(t)}, p^{(t)}) \,|\, \lambda(1:t-1) \big] ~.
	\]
	
	Note that $\E[u_t|\lambda(1:t-1)]=0$, so $\sum u_t$ is a martingale so we have that with high probability
	\[
	\textstyle
	\sum_{t \in [n]} u_t \ge - \tilde{O}(\sqrt{n} p_{max}) ~.
	\]
	
	Summing over $t \in [n]$ together with Eqn.~\eqref{eqn:online-lagrangian-bound} proves the lemma.
\end{proof}

Putting together Lemma~\ref{lem:online-one-step-2} and Lemma~\ref{lem:online-lagrangian-bound} proves Lemma~\ref{lem:online-key}.

\subsubsection{Approximate optimality}

Let $p_{m+1}=p_{max}$ and $p_{j}=0$ for all $j\leq m$.
By Lemma~\ref{lem:online-key}, the following holds with high probability:
\begin{align*}
\textstyle
 OPT - ALG & = OPT - L(x,p) \\
&
\textstyle
\leq
\frac{1}{\eta} D_{KL}(p \| p^{(1)}) + \eta \cdot \tilde{O} \big( n p_{\max} \sigma^2 \big) + \tilde{O} \big(\sqrt{n} p_{\max} \sigma \big) \\
&
\textstyle
=\frac{1}{\eta}p_{\max}\ln(m+1) + \eta \cdot \tilde{O} \big( n p_{\max} \sigma^2 \big) + \tilde{O} \big(\sqrt{n} p_{\max} \sigma \big)~.
\end{align*}
Plug in our choice of parameters, the RHS further equals
\[
\tilde{O} \big( \sqrt{n} p_{\max} \sigma \big) = O(\alpha n) ~.
\]
So we have the desired approximate optimality guarantee.

\subsubsection{Approximate feasibility}

We choose $p$ to penalize the over-demands and, thus, make $L(x^{(t)},p)$ as small as possible.
We let $p_{j^{*}}=p_{\max}$, where $j^{*}$ is the most over-demanded constraint and let $p_{j}=0$ for any $j\neq j^{*}$.
Let $s = b - \sum_{i\in[n]}\sum_{k\in[\ell_i]} a_{ij^*k} x_{ik}$ be the over-demand of $j^*$.
By Lemma~\ref{lem:online-key} and the choice of $p$ and $p^{(1)}$, with high probability we have:
\begin{equation}
\label{eqn:online-feasibility-1}
OPT - L(x, p) \le \tilde{O} \big( \sqrt{n} p_{\max} \sigma \big)
\end{equation}
By the choice of $p$, we further get that:
%
\[
\textstyle
L(x,p) = ALG - p_{\max} s ~.
\]
%
Note that $ALG \leq (1+\frac{s}{b})OPT$, because increasing the supply per resource from $b$ to $b+s$ increases the optimal packing objective by at most a $\frac{b+s}{b}$ factor.
So we have:
\[
\textstyle
OPT - L(x, p)
\geq OPT-(1+\frac{s}{b})OPT+p_{\max} s = s \left(p_{\max}-\frac{1}{b} OPT\right) \geq\frac{sp_{max}}{2} ~,
\]
where the last inequality is due to $p_{\max} = \tilde{O}(\frac{\alpha \sqrt{n}}{\sigma})$, $b \ge \tilde{O} \big( \frac{\sqrt{n} \sigma}{\alpha} \big)$, and $OPT \le n$.
So we have that:
\[
\textstyle
\frac{sp_{max}}{2}
\leq \tilde{O} \big( \sqrt{n} p_{\max} \sigma \big) ~.
\]
So we have $s \le \tilde{O}(\sqrt{n} \sigma)$.
Recall that $b \ge \tilde{O} \big( \frac{\sqrt{n}\sigma}{\alpha} \big)$,
we have $s \le \alpha b$.

%% file: appendix.tex
\section{Missing proofs in Section~\ref{sec:algorithm}}
\label{sec:concentration}

\subsection{Proof of Lemma~\ref{lem:3}}
\begin{proof}
	By the definition of KL-divergence, we get that:
	\begin{align}
	\label{eq:1}
	\begin{split}
	D_{KL}(p \| p^{(t+1)})-D_{KL}(p \| p^{(t)})
	&=
	\textstyle
	\sum_{j=1}^{m+1}p_j\ln \big( \tfrac{p^{(t)}}{p^{(t+1)}} \big)\\
	&=
	\textstyle
	\sum_{j=1}^{m+1} p_j \ln \big( \tfrac{\phi^{(t)}}{1-\eta\nabla_j\bar{D}(p^{(t)})} \big)\\
	&=
	\textstyle
	p_{\max}\ln \phi^{(t)} + \sum_{j=1}^{m+1} p_j \ln \big( \tfrac{1}{1-\eta\nabla_j\bar{D}(p^{(t)})} \big) ~,
	\end{split}
	\end{align}
	where the last equality is due to $ \| p \|_1 = p_{\max}$.
	Next, we bound the two terms separately.
	The first term equals:
	\begin{align*}
	p_{\max}\ln\phi^{(t)}
	&
	\textstyle
	=
	p_{\max}\ln \big( \tfrac{1}{p_{\max}}\sum_{j=1}^mp_j^{(t)} \big( 1-\eta\nabla_j\bar{D}(p^{(t)}) \big) \big) \\
	&
	\textstyle
	=
	p_{\max} \ln \big(1-\tfrac{\eta}{p_{\max}}\big\langle p^{(t)},\nabla \bar{D}(p^{(t)}) \big\rangle \big) ~.
	\end{align*}
	By the definition of $\nabla \bar{D}(p^{(t)})$, we have that
	\[
	\tfrac{\eta}{p_{\max}} \big\langle p^{(t)},\nabla \bar{D}(p^{(t)}) \big\rangle  \le \tfrac{\eta}{p_{\max}} \big\langle p^{(t)}, \nabla_{\max} \cdot \mathbf{1} \big \rangle
 = \tfrac{\eta}{p_{\max}} \cdot \nabla_{\max} \| p^{(t)} \|_1
 = \eta \nabla_{\max} < 1 ~.
	\]
	and similarly $\tfrac{\eta}{p_{\max}} \big\langle p^{(t)},\nabla \bar{D}(p^{(t)}) \big\rangle \ge - \tfrac{1}{2}$.
	Note that $\ln(1-x) \le -x$ for any $x \le 1$, we have:
	\begin{equation}
	\textstyle
	\label{eq:2}
	p_{\max}\ln\phi^{(t)} \le -\eta \big\langle p^{(t)},\nabla \bar{D}(p^{(t)}) \big\rangle ~.
	\end{equation}
	Now we bound the second term using inequalities $\ln(\frac{1}{1-xy})\leq \ln(\frac{1}{1-x})y$ for any $0 \le x, y \le 1$ and $\ln(\frac{1}{1-xy})\leq \ln(1+x)y$ for any $0 \le x \le 1, -1\leq y \leq 0$.
	\begin{align*}
	&
	\textstyle
	\sum_{j=1}^{m+1}p_j\ln\left(\frac{1}{1-\eta\nabla_j \bar{D}(p^{(t)})}\right) \\
	& \quad\quad \leq
	\textstyle
	\sum_{j:\nabla_j\geq 0} p_j \ln \left(\frac{1}{1-\eta \nabla_{\max}}\right) \cdot \frac{\nabla_j\bar{D}(p^{(t)})}{\nabla_{\max}} + \sum_{j:\nabla_j<0} p_j \ln\left(1+\eta \nabla_{\max}\right) \cdot \frac{\nabla_j\bar{D}(p^{(t)})}{\nabla_{\max}}
	\end{align*}
	Then, we further upper bound the above using $\ln(\frac{1}{1-x})\leq x+x^2$ and $\ln(1+x) \ge x - x^2$, and get that:
	\begin{align*}
	\textstyle
	 \sum_{j=1}^{m+1}p_j\ln\left(\frac{1}{1-\eta\nabla_j \bar{D}(p^{(t)})}\right) & \leq
	\textstyle
	\sum_{j:\nabla_j\geq 0}\left(\eta p_{j} \nabla \bar{D}(p^{(t)})+\eta^2 \nabla_{\max}p_j\nabla_j\bar{D}(p^{(t)})\right) \\
	& \quad \quad
	\textstyle
	+ \sum_{j:\nabla_j < 0}\left(\eta p_{j} \nabla \bar{D}(p^{(t)})-\eta^2 \nabla_{\max}p_j\nabla_j\bar{D}(p^{(t)})\right)\\
	&
	\textstyle
	=
	\eta\big\langle p,\nabla \bar{D}(p^{(t)})\big\rangle+\eta^2\nabla_{\max} \big\langle p, |\nabla \bar{D}(p^{(t)})| \big\rangle ~.
	\end{align*}
	Finally, by that $\| p \|_1 \le p_{\max}$, and $\nabla_j \bar{D}(p^{(t)}) \le \nabla_{\max}$ for any $j$, $\big\langle p, |\nabla \bar{D}(p^{(t)})| \big\rangle$ is upper bounded by $p_{\max} \nabla_{\max}$.
	So we have:
	\begin{equation}
	\label{eq:3}
\begin{split}
	\textstyle
	\sum_{j=1}^{m+1}p_j\ln\left(\frac{1}{1-\eta\nabla_j \bar{D}(p^{(t)})}\right) & \le \eta \big\langle p,\nabla \bar{D}(p^{(t)}) \big\rangle  + \eta^2 p_{\max} \nabla_{\max}^2 ~.
\end{split}	
\end{equation}
	Putting \eqref{eq:1}, \eqref{eq:2}, and \eqref{eq:3} together proves the lemma.
\end{proof}

\subsection{Proof of Lemma \ref{lem:4}}
\label{prolem4}
\begin{proof}
In our case, $v^{(t)}_j$'s are unbounded but have exponentially small tail contributions.
We follow the standard strategy of proving Azuma-Hoeffding type of concentration bounds.
By symmetry of the random variables $v^{(t)}_j$'s, it suffices to show the first inequality.

Let $X_k = \sum_{t=0}^k \big\langle q^{(t)}, v^{(t)} \big\rangle$ for $0 \le k \le T$.
For a positive $\lambda < \frac{\epsilon^{\prime}}{p_{\max}}$ whose value will be determined later, we have:
\begin{align*}
&
\textstyle
\Pr \left[ ~ X_T - X_0 \ge \frac{p_{\max}\sqrt{8 T \ln(\frac{6}{\beta})}}{\epsilon^{\prime}} ~ \right] \\
&
\textstyle\quad\quad
=
\Pr \left[ \exp{ \left( \lambda \left( X_T - X_0 - \frac{p_{\max}\sqrt{8 T \ln(\frac{6}{\beta})}}{\epsilon^{\prime}} \right) \right)} \ge 1 \right] \\
&
\textstyle\quad\quad
\le
\E \left[ \exp{ \left( \lambda \left( X_T - X_0 - \frac{p_{\max}\sqrt{8 T \ln(\frac{6}{\beta})}}{\epsilon^{\prime}} \right) \right)} \right] ~.
\end{align*}
Next, we upper bound $\E \big[ \exp \big( \lambda \big( X_T - X_0 \big) \big) \big]$, which can be rewritten as:
\[
\begin{split}
& \E_{\nu^{(1)}} [\exp\big( \lambda (X_1 - X_0) \big) \cdot \E_{\nu^{(2)} \,|\, \nu{(1)}} [\exp\big( \lambda (X_2 - X_1) \big) \cdot \dots \E_{\nu^{(T)} \,|\, \nu^{(1:T-1)}}[ \exp \big(\lambda(X_T - X_{T-1})\big) ]\dots ]]~.
\end{split}
\]
For any $1 \le t \le T$, we have that:
\begin{align*}
 \textstyle
\E_{\nu^{(t)} \,|\, \nu^{(1:t-1)}} \exp \big( \lambda \big( X_t - X_{t-1} \big) \big)
=
\prod_{j=1}^m \E \big[ \exp{\big( \lambda q^{(t)}_j v^{(t)} \big)} \big] ~.
\end{align*}
Further, for any $1 \le j \le m$, we have:
\begin{align*}
~
\E \big[ \exp{\big( \lambda q^{(t)}_j v^{(t)} \big)} \big]
& =
\textstyle
\int_{-\infty}^{+\infty} \tfrac{\epsilon^{\prime}}{2} \exp{(- \epsilon^{\prime} |y|)} \exp{(\lambda q^{(t)}_j y)} dy \\
&  =
\textstyle
\int_0^{+\infty} \tfrac{\epsilon^{\prime}}{2} \exp{(- \epsilon^{\prime} y)} \exp(\lambda q^{(t)}_j y) dy \\
& \quad + \int_0^{+\infty} \tfrac{\epsilon^{\prime}}{2} \exp{(- \epsilon^{\prime} y)} \exp{(- \lambda q^{(t)}_j y)} dy\\
&  =
\textstyle \frac{\epsilon^{\prime}}{2(\epsilon^{\prime} - \lambda q^{(t)}_j)} + \frac{\epsilon^{\prime}}{2(\epsilon^{\prime} + \lambda q^{(t)}_j)} \\
&  =
\textstyle \frac{1}{1 - (\lambda q^{(t)}_j / \epsilon^{\prime})^2} < 1 + 2 (\lambda q^{(t)}_j / \epsilon^{\prime})^2 < \exp{\big( 2 (\lambda q^{(t)}_j / \epsilon^{\prime})^2 \big)} ~.
\end{align*}
The last second inequality holds for any $\lambda \leq \frac{\epsilon^{'}}{2 q_{j}^{(t)}}$.
Hence, we have:
\begin{align*}
\textstyle
\E_{\nu^{(t)} \,|\, \nu^{(1:t-1)}} \exp \big( \lambda \big( X_t - X_{t-1} \big) \big)
\textstyle
\le
\exp{\big( \sum_{j = 1}^m 2 (\lambda q^{(t)}_j / \epsilon^{\prime})^2 \big)} ~.
\end{align*}
Next, note that $\sum_{j=1}^m (q^{(t)}_j)^2 \le \big( \sum_{j=1}^m q^{(t)}_j \big)^2 = p_{\max}^2$.
We have:
\[
\textstyle
\E_{\nu^{(t)} \,|\, \nu^{(1:t-1)}} \exp \big( \lambda \big( X_t - X_{t-1} \big) \big)
\le
\exp{\big( 2 (\lambda p_{\max} / \epsilon^{\prime})^2 \big)} ~.
\]
and, thus,
\[
\textstyle
\E \big[ \exp \big( \lambda \big( X_T - X_0 \big) \big) \big]
\le
\exp{\big( 2 T (\lambda p_{\max} / \epsilon^{\prime})^2 \big)} ~.
\]
Thus, we get that:
\begin{align*}
\begin{split}
 \textstyle
\E \left[ \exp \left( \lambda \left( X_T - X_0 - \frac{p_{\max}\sqrt{8 T \ln(\frac{6}{\beta})}}{\epsilon^{\prime}} \right) \right) \right]
\le
\exp{\left(\lambda^2 \cdot 2 T \left(\tfrac{p_{\max}}{\epsilon^{\prime}}\right)^2 - \lambda \cdot \frac{p_{\max}\sqrt{8 T \ln(\frac{6}{\beta})}}{\epsilon^{\prime}} \right)} ~.
\end{split}
\end{align*}
The lemma then follows by choosing
\[
\textstyle
\lambda = \frac{\epsilon^{\prime}}{p_{\max}} \cdot \sqrt{\frac{\ln(\frac{6}{\beta})}{2T}} ~.
\]
\end{proof}
\subsection{ Proof of Lemma \ref{lem:5}}
\label{prolem5}
\begin{proof}
By symmetry, it suffices to show the first inequality by symmetry of the random variables $v^{(t)}_j$'s.
For a positive $\lambda $ to be determined later, we have:
\begin{align*}
&
\textstyle
\Pr \left[ ~ \sum_{t=1}^T \sum_{j=1}^m q^{(t)}_j \cdot \max \big\{0, v^{(t)}_j - \tfrac{\ln(T)}{\epsilon^{\prime}} \big\}
 \quad \ge \frac{2 p_{\max}}{\epsilon^{\prime}} \ln(\frac{6}{\beta}) ~ \right] \\
&
\textstyle \qquad
=
\Pr \left[ ~ \exp{\left( \lambda \left( \sum_{t=1}^T \sum_{j=1}^m q^{(t)}_j \cdot \max \big\{0, v^{(t)}_j - \tfrac{\ln(T)}{\epsilon^{\prime}} \big\}  \right) \right)}
\cdot \exp{\left( - \frac{2 \lambda p_{\max}}{\epsilon^{\prime}} \ln(\frac{6}{\beta}) \right)} \ge 1 ~ \right] \\
&
\textstyle \qquad
\le
\E \left[ ~ \exp{\left( \lambda \left( \sum_{t=1}^T \sum_{j=1}^m q^{(t)}_j \cdot \max \big\{0, v^{(t)}_j  - \tfrac{\ln(T)}{\epsilon^{\prime}} \big\} \right) \right)} ~
\cdot \exp{\left( - \frac{2 \lambda p_{\max}}{\epsilon^{\prime}} \ln(\frac{6}{\beta}) \right)}\right]\\
&
\textstyle \qquad
= \prod_{t=1}^T \prod_{j=1}^m E \left[ ~ \exp{\left( \lambda \left( q^{(t)}_j \cdot \max \big\{0, v^{(t)}_j - \tfrac{\ln(T)}{\epsilon^{\prime}} \big\} \right) \right) } | q^{(1)}, \dots q^{(t-1)} \right] \cdot \exp{\left( - \frac{2 \lambda p_{\max}}{\epsilon^{\prime}} \ln(\frac{6}{\beta}) \right)} ~.
\end{align*}
%
Next, we bound
 $E \left[ ~ \exp{\left( \lambda \left( q^{(t)}_j \cdot \max \big\{0, v^{(t)}_j - \tfrac{\ln(T)}{\epsilon^{\prime}} \big\} \right) \right) }  \right]$ by:
\begin{align*}
&
\textstyle
E \left[ ~ \exp{\left( \lambda \left( q^{(t)}_j \cdot \max \big\{0, v^{(t)}_j - \tfrac{\ln(T)}{\epsilon^{\prime}} \big\} \right) \right) }  \right] \\
&
\textstyle \qquad
=\int_{\frac{\ln(T)}{\epsilon^{\prime}}}^{+\infty} \tfrac{\epsilon^{\prime}}{2} \exp{\big(- \epsilon^{\prime} y \big)} \exp{\big( \lambda q^{(t)}_j \cdot (y - \frac{\ln(T)}{\epsilon^{\prime}}) \big)} dy+
\Pr \big[ v^{(t)}_j \le \frac{\ln(T)}{\epsilon^{\prime}} \big] \\
\textstyle
& \qquad
=
1 - \tfrac{1}{2T}
+
\tfrac{1}{2T} \cdot \tfrac{\epsilon^{\prime}}{\epsilon^{\prime} - \lambda q^{(t)}_j}
=
1 + \tfrac{1}{2T} \tfrac{\lambda q^{(t)}_j}{\epsilon^{\prime} - \lambda q^{(t)}_j}
\le
1 + \tfrac{1}{T} \tfrac{\lambda q^{(t)}_j}{\epsilon^{\prime}}
\le
\exp{\big( \tfrac{1}{T} \tfrac{\lambda q^{(t)}_j}{\epsilon^{\prime}} \big)} ~,
\end{align*}
where the second last inequality is due to $\lambda q^{(t)}_j < \lambda p_{\max} < \tfrac{\epsilon^{\prime}}{2}$.
Thus, we have:
\begin{align*}
 \textstyle
\prod_{j=1}^m E \bigg[ ~ \exp{\left( \lambda \left(q^{(t)}_j \cdot \max \big\{0, v^{(t)}_j - \tfrac{\ln(T)}{\epsilon^{\prime}} \big\} \right) \right) } ~|~ q^{(1)}, \dots q^{(t-1)} ~ \bigg]
&
\textstyle
<
\exp{\big( \sum_{j=1}^m \tfrac{1}{T} \tfrac{\lambda q^{(t)}_j}{\epsilon^{\prime}} \big)} \\
&
\textstyle
\le
\exp{\big( \tfrac{1}{T} \tfrac{\lambda p_{\max}}{\epsilon^{\prime}} \big)} \\
&
\textstyle
\le
\exp{\big( \frac{1}{2T} \big)} ~.
\end{align*}
So we have:
\[
\begin{split}
&\textstyle
\prod_{t=1}^T \prod_{j=1}^m E \left[ ~ \exp{\left( \lambda \left( q^{(t)}_j \cdot \max \big\{0, v^{(t)}_j - \tfrac{\ln(T)}{\epsilon^{\prime}} \big\} \right) \right) } | q^{(1)}, \dots q^{(t-1)} \right]
\le
\exp{\big( \frac{1}{2} \big)} < 2 ~.
\end{split}
\]
and, thus,
\begin{align*}
\textstyle
\Pr \left[ ~ \exp{\left( \lambda \left( \sum_{t=1}^T \sum_{j=1}^m q^{(t)}_j \cdot \max \big\{0, v^{(t)}_j - \tfrac{\ln(T)}{\epsilon^{\prime}} \big\} - \frac{2}{\epsilon^{\prime}} \ln(\frac{6}{\beta}) \right) \right)}
 \ge 1 ~ \right]
\le
2 \cdot \exp{\left( - \frac{2 \lambda p_{\max}}{\epsilon^{\prime}} \ln(\frac{6}{\beta}) \right)} ~,
\end{align*}
which equals $\frac{\beta}{3}$ due to our choice of $\lambda = \frac{\epsilon^{\prime}}{2 p_{\max}}$.
\end{proof}

\subsection{Proof of Lemma~\ref{lem2}}

\begin{proof}
	By Lemma~\ref{lem:1}, we have:
	\begin{align*}
	\textstyle
	\sum_{t=1}^T \left( L(x^{(t)},p^{(t)}) - L(x^{(t)},p) \right)
	& \textstyle\quad
	= \sum_{t=1}^T \big\langle p^{(t)} - p, \nabla D(p^{(t)}) \big\rangle \\
	& \textstyle\quad
	= \sum_{t=1}^T \big\langle p^{(t)} - p, \nabla \bar{D}(p^{(t)}) \big\rangle
	\\
	& \quad\quad\textstyle
	+ \sum_{t=1}^T\big\langle p^{(t)} - p, \nabla \hat{D}(p^{(t)}) - \nabla \bar{D}(p^{(t)}) \big\rangle \\
	& \quad\quad\textstyle
	+ \sum_{t=1}^T\big\langle p^{(t)} - p, \nabla D(p^{(t)}) - \nabla \hat{D}(p^{(t)}) \big\rangle ~.
	\end{align*}
	Then, applying  \ref{lem:3}, we further get that:
	\begin{align*}
	 \textstyle
	\sum_{t=1}^T \big\langle p - p^{(t)}, \nabla \bar{D}(p^{(t)}) \big\rangle
	& \leq \sum_{t=1}^T ( \frac{1}{\eta} \big( D_{KL}(p \| p^{(t)}) - D_{KL}(p \| p^{(t+1)}) \big) + \eta \cdot p_{\max} \nabla_{\max}^2 ) \\
	&
	\textstyle
	=
	\frac{1}{\eta} \big( D_{KL}(p \| p^{(1)}) - D_{KL}(p \| p^{(T+1)}) \big)
	+ \eta \cdot T p_{\max} \nabla_{\max}^2 \\
	&
	\textstyle
	\le
	\frac{1}{\eta} D_{KL}(p \| p^{(1)}) + \eta \cdot T p_{\max} \nabla_{\max}^2 ~.
	\end{align*}
	Putting together, we have:
	\begin{align}
	\label{eq:4}
	\begin{split}
	 \textstyle
	\sum_{t=1}^T \left( L(x^{(t)},p^{(t)}) - L(x^{(t)},p) \right)
	&\leq \frac{1}{\eta} D_{KL}(p \| p^{(1)}) +\eta \cdot T p_{\max} \nabla_{\max}^2 \\
	& \quad\quad\textstyle
	+ \sum_{t=1}^T\big\langle p^{(t)} - p, \nabla \hat{D}(p^{(t)}) - \nabla \bar{D}(p^{(t)}) \big\rangle \\
	& \quad\quad\textstyle
	+ \sum_{t=1}^T\big\langle p^{(t)} - p, \nabla D(p^{(t)}) - \nabla \hat{D}(p^{(t)}) \big\rangle ~.
	\end{split}
	\end{align}
	Next, we bound the last two terms separately.
	By the definition of $\nabla \hat{D}(p^{(t)})$, the last term can be rewritten as:
	\begin{align*}
	& \textstyle
	\sum_{t=1}^T\big\langle p^{(t)} - p, \nabla D(p^{(t)}) - \nabla \hat{D}(p^{(t)}) \big\rangle 
	=
	\sum_{t=1}^T\big\langle p - p^{(t)}, \nu^{(t)} \big\rangle ~.
	\end{align*}
	Applying Lemma \ref{lem:4} twice, we get that with probability at least $1 - \tfrac{\beta}{3}$,
	\begin{equation}
	\label{eq:5}
    \begin{split}
	 \textstyle
	\sum_{t=1}^T\big\langle p^{(t)} - p, \nabla D(p^{(t)}) - \nabla \hat{D}(p^{(t)}) \big\rangle
	\le
	\frac{2 p_{\max}\sqrt{8 T \ln(\frac{6}{\beta})}}{\epsilon^{\prime}} 
	=
	\frac{16 p_{\max} T \sqrt{m \ln(\frac{6}{\beta}) \ln(\frac{2}{\delta})}}{\epsilon} ~.
    \end{split}
	\end{equation}
	%
	\\
	As for the second last term, we can upper bound it as:
	\[
\begin{split}
	 \textstyle
	\sum_{t=1}^T\big\langle p^{(t)} - p, \nabla \hat{D}(p^{(t)}) - \nabla \bar{D}(p^{(t)}) \big\rangle \le
	\sum_{t=1}^T \sum_{j=1}^m \big| p^{(t)}_j - p_j \big| \cdot \big| \nabla_j \hat{D}(p^{(t)}) - \nabla_j \bar{D}(p^{(t)}) \big| ~.
\end{split}	
\]
	Note that by the definition of $\nabla_j \bar{D}(p^{(t)})$, $\big| \nabla_j \hat{D}(p^{(t)}) - \nabla_j \bar{D}(p^{(t)}) \big|$ can be rewritten as:
	\[
	 \max \big\{0, \nabla_j D(p^{(t)}) + \nu^{(t)}_j - \nabla_{\max} \big\} + \max \big\{0, - \nabla_j D(p^{(t)}) - \nu^{(t)}_j - \nabla_{\max} \big\} ~.
\]
	Further, $\nabla_j D(p^{(t)}) = b - \sum_{i\in[n]}\sum_{k\in [\ell_i]}a_{ij}x_{ik}^{(t)} \in [-n, b] \subseteq [-n, n]$.
	So the above is at most (recall that $\nabla_{\max} = n + \tfrac{\ln(T)}{\epsilon^{\prime}}$):
	\begin{align*}
	&
	\textstyle
	\max \big\{0, n + \nu^{(t)}_j - \nabla_{\max} \big\} + \max \big\{0, n - \nu^{(t)}_j - \nabla_{\max} \big\} \\
	&
	\textstyle\quad\quad
	=
	\max \big\{0, \nu^{(t)}_j - \tfrac{\ln(T)}{\epsilon^{\prime}} \big\} + \max \big\{0, - \nu^{(t)}_j - \tfrac{\ln(T)}{\epsilon^{\prime}} \big\} ~.
	\end{align*}
	Applying Lemma~\ref{lem:5} twice, we get that with probability at least $1-\tfrac{2\beta}{3}$,
	\begin{align}
	\begin{split}
	\label{eq:6}
	&
	\textstyle
	\sum_{t=1}^T\big\langle p^{(t)} - p, \nabla \hat{D}(p^{(t)}) - \nabla \bar{D}(p^{(t)}) \big\rangle \\
	&
	\textstyle\quad\quad
	\le
	\sum_{t=1}^T \sum_{j=1}^m \big| p^{(t)}_j - p_j \big| \cdot ( \max \big\{0, \nu^{(t)}_j - \tfrac{\ln(T)}{\epsilon^{\prime}} \big\}  + \max \big\{0, - \nu^{(t)}_j - \tfrac{\ln(T)}{\epsilon^{\prime}} \big\} ) \\
	&
	\textstyle\quad\quad
	\le
	\frac{4 p_{\max}}{\epsilon^{\prime}} \ln(\frac{6}{\beta}) \le
	\frac{4 p_{\max} T \sqrt{m \ln(\frac{6}{\beta}) \ln(\frac{2}{\delta})}}{\epsilon} ~.
	\end{split}
	\end{align}
	Putting together \eqref{eq:4}, \eqref{eq:5}, and \eqref{eq:6} proves the lemma.
\end{proof}

%% file: main.bbl
\begin{thebibliography}{10}

\bibitem{DBLP:conf/soda/AgrawalD15}
Shipra Agrawal and Nikhil~R. Devanur.
\newblock Fast algorithms for online stochastic convex programming.
\newblock In {\em {SODA}}, pages 1405--1424. {SIAM}, 2015.

\bibitem{azuma1967weighted}
Kazuoki Azuma.
\newblock Weighted sums of certain dependent random variables.
\newblock {\em Tohoku Mathematical Journal, Second Series}, 19(3):357--367,
  1967.

\bibitem{DBLP:conf/focs/BassilyST14}
Raef Bassily, Adam~D. Smith, and Abhradeep Thakurta.
\newblock Private empirical risk minimization: Efficient algorithms and tight
  error bounds.
\newblock In {\em {FOCS}}, pages 464--473. {IEEE} Computer Society, 2014.

\bibitem{DBLP:journals/jacm/BlumLR13}
Avrim Blum, Katrina Ligett, and Aaron Roth.
\newblock A learning theory approach to noninteractive database privacy.
\newblock {\em J. {ACM}}, 60(2):12:1--12:25, 2013.

\bibitem{buchbinder2009online}
Niv Buchbinder and Joseph Naor.
\newblock Online primal-dual algorithms for covering and packing.
\newblock {\em Mathematics of Operations Research}, 34(2):270--286, 2009.

\bibitem{DBLP:journals/jmlr/ChaudhuriH11}
Kamalika Chaudhuri and Daniel~J. Hsu.
\newblock Sample complexity bounds for differentially private learning.
\newblock In {\em {COLT}}, volume~19 of {\em {JMLR} Proceedings}, pages
  155--186. JMLR.org, 2011.

\bibitem{DBLP:conf/tcc/DworkMNS06}
Cynthia Dwork, Frank McSherry, Kobbi Nissim, and Adam~D. Smith.
\newblock Calibrating noise to sensitivity in private data analysis.
\newblock In {\em {TCC}}, volume 3876 of {\em Lecture Notes in Computer
  Science}, pages 265--284. Springer, 2006.

\bibitem{DBLP:conf/focs/DworkRV10}
Cynthia Dwork, Guy~N. Rothblum, and Salil~P. Vadhan.
\newblock Boosting and differential privacy.
\newblock In {\em {FOCS}}, pages 51--60. {IEEE} Computer Society, 2010.

\bibitem{dwork2010differential}
Cynthia Dwork and Adam Smith.
\newblock Differential privacy for statistics: What we know and what we want to
  learn.
\newblock {\em Journal of Privacy and Confidentiality}, 1(2):2, 2010.

\bibitem{ghosh2014buying}
Arpita Ghosh, Katrina Ligett, Aaron Roth, and Grant Schoenebeck.
\newblock Buying private data without verification.
\newblock In {\em EC}, pages 931--948. ACM, 2014.

\bibitem{DBLP:conf/stoc/HsuHRRW14}
Justin Hsu, Zhiyi Huang, Aaron Roth, Tim Roughgarden, and Zhiwei~Steven Wu.
\newblock Private matchings and allocations.
\newblock In {\em {STOC}}, pages 21--30. {ACM}, 2014.

\bibitem{DBLP:conf/soda/Hsu0RW16}
Justin Hsu, Zhiyi Huang, Aaron Roth, and Zhiwei~Steven Wu.
\newblock Jointly private convex programming.
\newblock In {\em {SODA}}, pages 580--599. {SIAM}, 2016.

\bibitem{DBLP:conf/icalp/HsuRRU14}
Justin Hsu, Aaron Roth, Tim Roughgarden, and Jonathan Ullman.
\newblock Privately solving linear programs.
\newblock In {\em {ICALP} {(1)}}, volume 8572 of {\em Lecture Notes in Computer
  Science}, pages 612--624. Springer, 2014.

\bibitem{DBLP:reference/algo/Huang16c}
Zhiyi Huang.
\newblock Privacy preserving auction.
\newblock In {\em Encyclopedia of Algorithms}, pages 1618--1622. 2016.

\bibitem{kasiviswanathan2011can}
Shiva~Prasad Kasiviswanathan, Homin~K Lee, Kobbi Nissim, Sofya Raskhodnikova,
  and Adam Smith.
\newblock What can we learn privately?
\newblock {\em SIAM Journal on Computing}, 40(3):793--826, 2011.

\bibitem{DBLP:conf/innovations/KearnsPRU14}
Michael Kearns, Mallesh~M. Pai, Aaron Roth, and Jonathan Ullman.
\newblock Mechanism design in large games: incentives and privacy.
\newblock In {\em {ITCS}}, pages 403--410. {ACM}, 2014.

\bibitem{DBLP:conf/focs/McSherryT07}
Frank McSherry and Kunal Talwar.
\newblock Mechanism design via differential privacy.
\newblock In {\em {FOCS}}, pages 94--103. {IEEE} Computer Society, 2007.

\bibitem{milgrom2002envelope}
Paul Milgrom and Ilya Segal.
\newblock Envelope theorems for arbitrary choice sets.
\newblock {\em Econometrica}, 70(2):583--601, 2002.

\bibitem{plotkin1995fast}
Serge~A Plotkin, David~B Shmoys, and {\'E}va Tardos.
\newblock Fast approximation algorithms for fractional packing and covering
  problems.
\newblock {\em Mathematics of Operations Research}, 20(2):257--301, 1995.

\bibitem{DBLP:conf/sigecom/RogersR14}
Ryan~M. Rogers and Aaron Roth.
\newblock Asymptotically truthful equilibrium selection in large congestion
  games.
\newblock In {\em {EC}}, pages 771--782. {ACM}, 2014.

\bibitem{DBLP:conf/sigecom/RogersRUW15}
Ryan~M. Rogers, Aaron Roth, Jonathan Ullman, and Zhiwei~Steven Wu.
\newblock Inducing approximately optimal flow using truthful mediators.
\newblock In {\em {EC}}, pages 471--488. {ACM}, 2015.

\bibitem{srinivasan1999improved}
Aravind Srinivasan.
\newblock Improved approximation guarantees for packing and covering integer
  programs.
\newblock {\em SIAM Journal on Computing}, 29(2):648--670, 1999.

\bibitem{steinke2017between}
Thomas Steinke and Jonathan Ullman.
\newblock Between pure and approximate differential privacy.
\newblock {\em Journal of Privacy and Confidentiality}, 7(2):2, 2017.

\end{thebibliography}
